\documentclass[10pt,journal,compsoc]{IEEEtran}
\ifCLASSOPTIONcompsoc
  \usepackage[nocompress]{cite}
\else
  \usepackage{cite}
\fi
\ifCLASSINFOpdf
\else
\fi
\usepackage{amsthm}

\usepackage{cite}
\usepackage{amsmath,amssymb,amsfonts}
\usepackage{algorithmic}
\usepackage{algorithm}
\usepackage{graphicx}
\usepackage{textcomp}
\usepackage{xcolor}
\newcommand{\RNum}[1]{\uppercase\expandafter{\romannumeral #1\relax}}
\DeclareMathOperator*{\argmax}{argmax}
\newtheorem{theorem}{Theorem}
\newtheorem{definition}[theorem]{Definition}
\usepackage[acronym]{glossaries}
\newacronym{RV}{r.v.}{random variable}
\newtheorem{lemma}{Lemma}

\usepackage{multicol,subfig}
\usepackage{dblfloatfix} 

\newtheorem{cor}{Corollary}

\begin{document}
%
\title{Privacy-Preserving Edge Caching: \\A Probabilistic Approach}
%
%
%
%

\author{Seyedeh~Bahereh~Hassanpour,
        Ahmad~Khonsari,
        ~Masoumeh~Moradian,
        and~Seyed~Pooya~Shariatpanahi
\IEEEcompsocitemizethanks{\IEEEcompsocthanksitem S. B. Hassanpour, A. Khonsari, and S. P. Shariatpanahi are with the Department
of Electrical and Computer Engineering, Tehran University, Tehran, Iran.\protect\\
E-mails: {b.hassanpour, a\_khonsari, p.shariatpanahi}@ut.ac.ir.
\IEEEcompsocthanksitem M. Moradian and A. Khonsari are with School of Computer, Institute for Research in Fundamental Sience (IPM), Tehran, Iran.\protect\\
E-mails:{mmoradian, ak}@ipm.ir
}
}

%
%

\markboth{Journal of \LaTeX\ Class Files,~Vol.~14, No.~8, August~2015}%
{Shell \MakeLowercase{\textit{et al.}}: Bare Demo of IEEEtran.cls for Computer Society Journals}
%



\IEEEtitleabstractindextext{%
\begin{abstract}
Edge caching (EC) decreases the average access delay of end-users through caching popular content at the edge of the network, however, it increases the leakage probability of valuable information such as users' preferences. Most of the existing privacy-preserving approaches focus on adding extra layers of encryption, which confronts the network with more challenges such as energy and computation limitations. We employ a chunk-based joint probabilistic caching (JPC) approach to mislead an adversary eavesdropping on the communication inside an EC and maximizing the adversary's error in estimating the requested file and the requesting cache. In JPC, we optimize the probability of each cache placement to minimize the communication cost while guaranteeing the desired privacy and then, formulate the optimization problem as a linear programming (LP) problem. Since JPC inherits the curse of dimensionality, we also propose scalable JPC (SPC), which reduces the number of feasible cache placements by dividing files into non-overlapping subsets. We also compare the JPC and SPC approach against an existing probabilistic method, referred to as disjoint probabilistic caching (DPC) and random dummy-based approach (RDA). Results obtained through extensive numerical evaluations confirm the validity of the analytical approach, the superiority of JPC and SPC over DPC and RDA. 
\end{abstract}

}

\maketitle

\IEEEdisplaynontitleabstractindextext

%
\IEEEpeerreviewmaketitle

\IEEEraisesectionheading{\section{Introduction}\label{sec:introduction}}

\IEEEPARstart{T}{he} deluge of interest in using delay-critical Internet-of-things (IoT) applications and delay-sensitive data services, even beyond fifth-generation (B5G) communications, poses a significant challenge to the data processing capabilities of the network. In this regard, the edge network emerged as a pivotal key in 5G for providing computation, storage, and processing much closer to the end-users compared to the cloud networks. Edge computing, as a paradigm in edge networks, provides service delivery at the edge of the network and mitigates transmitting the computation to more distant servers inside the cloud through equipping the intermediate servers with edge nodes such as micro base stations or WiFi access points \cite{ xiong2020resource,liu2022distributed}. Furthermore, edge caching (EC), as a promising technique to store popular content closer to the end-users, decreases the traffic of the backhaul links and improves the quality of experience (QoE) by the end-users in terms of access delay \cite{yousefpour2019all,khan2019edge}.

EC paradigm enhances security and privacy by bringing the content closer to the end-users and eliminating the access of multiple intermediate nodes to the data. However, this geographical proximity also brings the potential (active/passive) attackers closer to the critical/personal information such as users’ location or personal preferences. Therefore, the network is more vulnerable to different types of attacks \cite{ranaweera2021survey,ni2020security}. For example, an attacker may interrupt the communication between the user and the caching edge device through a jamming attack, break down an EC server through a distributed denial of service (DoS) attack, or get access to the caching contents and network resources as a fake legal user through spoofing attacks \cite{xiao2018security,cui2018multi}.
Moreover, the EC network consists of distributed edge devices controlled by autonomous people or companies. These owners may be curious about the data contents stored on their caches and even launch insider attacks or eavesdrop to obtain critical private information of the customers and sell them for different purposes, e.g., to the advertisements companies \cite{ xiao2018security, he2019physical}. Generally, the main motivation of privacy attacks in EC networks is to derive the identity of the requesting users, their queries, and the statistics of the queries, e.g., the popularity of the contents.

Much of the focus of researchers in recent years has been on studying the location \cite{ko2020lpga, jiang2021location, niu2014privacy} and the pattern privacies \cite{coopamootoo2020usage }, which aim to secure the location and the usage pattern of the users, respectively. The proposed solutions to tackle the privacy issues exploit cryptography and anonymity \cite{andreoletti2019privacy,andreoletti2019privacy1,acs2017privacy,cui2020edge}, information-theory \cite{schlegel2022privacy,hassanpour2020context}, machine-learning \cite{xiao2018security,yu2020mobility}, and dummy transmissions\cite{niu2014privacy}. In \cite{andreoletti2019privacy}, the authors propose a pseudonyms-based approach to conceal the real identity of the contents belonging to the content providers (CPs) and users' requests from the ISP, which is the cache owner, in a content distribution network (CDN). Due to the fact that over time ISP can discover the relation between the fake and real ID of the contents, there is a need to refresh the encrypted name. For this purpose, the authors in \cite{andreoletti2019privacy1} derived the optimal number of encryption refreshes in a static time (e.g., a day). In \cite{acs2017privacy}, the authors preserve the privacy of CPs and users from ISP by using Shamir secret sharing (SSS), which shares the content popularity among caches without letting the ISP know. Despite the abundance of cryptographic solutions, these solutions mostly suffer from computation complexity, power consumption, and decryption delays. However, the solutions proposed for providing security and privacy at the edge should possess low complexity due to the edge nodes' limited computational power and memory of the edge nodes and the energy and hardware limitations. The authors in \cite{schlegel2022privacy} propose a coding scheme that includes SSS and replicated subtasks to provide information-theoretic data privacy in the presence of untrustworthy edge servers. In \cite{hassanpour2020context} the authors applied information theory to maximize the lower bound for the best adversary's estimation error using Fano inequality. They mathematically formulate an $\epsilon$-constraint optimization model to find the probability of catching each file to maximize the adversary's error. In the above studies, the popularity of content is known. In the case of unknown or time-varying popularities, machine learning techniques are mostly employed to learn popularities. However, they require sharing the requests' information with a central node for the aim of training, which leads to privacy leakage. Federated learning is proposed to overcome the privacy issues in online learning scenarios \cite{yu2020mobility,yu2021privacy}. Another general policy in preserving privacy in the presence of eavesdroppers is the dummy-transmission-based approach, which relies on transmitting dummy queries by the caches or dummy information by the server\cite{niu2014privacy} in the content delivery phase in order to obfuscate the eavesdropper. The dummy-based approaches suffer from backhaul traffic increases due to extra unnecessary transmissions.

Probabilistic caching can also be employed in an EC network to leverage the privacy degree of the network since it increases the ambiguity of the adversaries over the caches' contents and thus, increases the privacy degree of the network. Furthermore, it takes advantage of lower complexity compared to cryptographic methods. Probabilistic cache placement is employed in \cite{shi2018probabilistic} to provide physical-layer security in a wireless cache-aided network in the presence of eavesdroppers. They optimize the probability of caching individual files to maximize the number of transmissions not decoded by eavesdroppers while minimizing the communication cost.

In this paper, we study probabilistic caching to preserve the desired privacy in an EC network while minimizing the network communication cost. In particular, our proposed EC network consists of a single server and $K$ distributed edge caches, where the server is in charge of content placement in the caches and delivering uncached contents to them. Furthermore, a passive eavesdropper monitors the amount of traffic transferred over the shared link between the server and the caches. In the proposed scenario, we minimize the communication cost while satisfying a minimum privacy degree, where the communication cost and privacy degree are defined as the average amount of traffic over the shared link and the error probability of the adversary, respectively. We assume that the adversary is aware of the content popularities and probabilistic caching protocols \cite{niu2015enhancing}. However, it has no access to the information of cache queries and their corresponding responses. As such, the adversary can only measure the amount of transferred traffic over the shared link in the content delivery (CD) phase and thus, exploits this information to estimate the identity of the requests. We optimize chunk-based probabilistic caching to increase the uncertainty of the adversary. Unlike the previous studies \cite{blaszczyszyn2015optimal,lin2019probabilistic}, which optimize the probability of caching individual files in order to satisfy the desired performance metrics, we optimize the probabilities of joint placements of the files in the caches and highlight its advantages throughout the paper. The main contributions of the paper are as follows.
\begin{itemize}
	\item We define joint probabilistic caching (JPC) policy rigorously and formulate the optimization of communication cost-constrained to guaranteed privacy under the JPC approach. Then, we show that the proposed optimization can be turned into a Linear Programming (LP) problem. We also propose the hit-ratio-based optimization of JPC and assert that chunk-based optimization provides more flexibility for achieving higher cache hit ratios.
	\item We solve the same optimization problem as in JPC considering the disjoint probabilistic caching (DPC) policy, in which the probabilities of caching individual files are optimized instead of the probabilities of different cache placements. Also, we show that optimal DPC has the same performance as optimized chunk-based JPC. However, JPC outperforms DPC when hit-ratio constraints and short-term performances are required.   
	\item We propose scalable JPC, in which the complexity decreases compared to JPC through caching chunks which are chosen from $L$ non-overlapping subsets of files instead of $N$ files. The performance of scalable JPC can be arbitrarily close to the optimal performance in the JPC through increasing the number of subsets $L$. 
	\item Finally, we present extensive numerical and simulation results to validate our analytical approach. We also propose a Random dummy approach as a benchmark and show that JPC and scalable JPC outperform the dummy approach significantly. Also, we show that by choosing proper subsets in scalable JPC, we can shrink the feasible set in the corresponding LP optimization significantly while keeping the performance very close to the optimal performance of JPC. 
\end{itemize}

The remainder of the paper is organized as follows. In Section \ref{sec:systemModel}, we describe the EC network, the adversary model, and the performance metrics. Section \ref{sec:problemformulation}, provides the problem formulation and LP optimization of JPC. Section \ref{sec:DPC} is dedicated to DPC optimization and its performance against JPC. Section \ref{sec:SPC} presents scalable JPC method. Numerical results are presented in Section \ref{sec:evaluation_simulation}. Finally, Section \ref{sec:conclusion} concludes the paper.

\section{System Model and Assumptions }
\label{sec:systemModel}

\begin{figure}
	\centering 
	\includegraphics[width=0.95\linewidth]{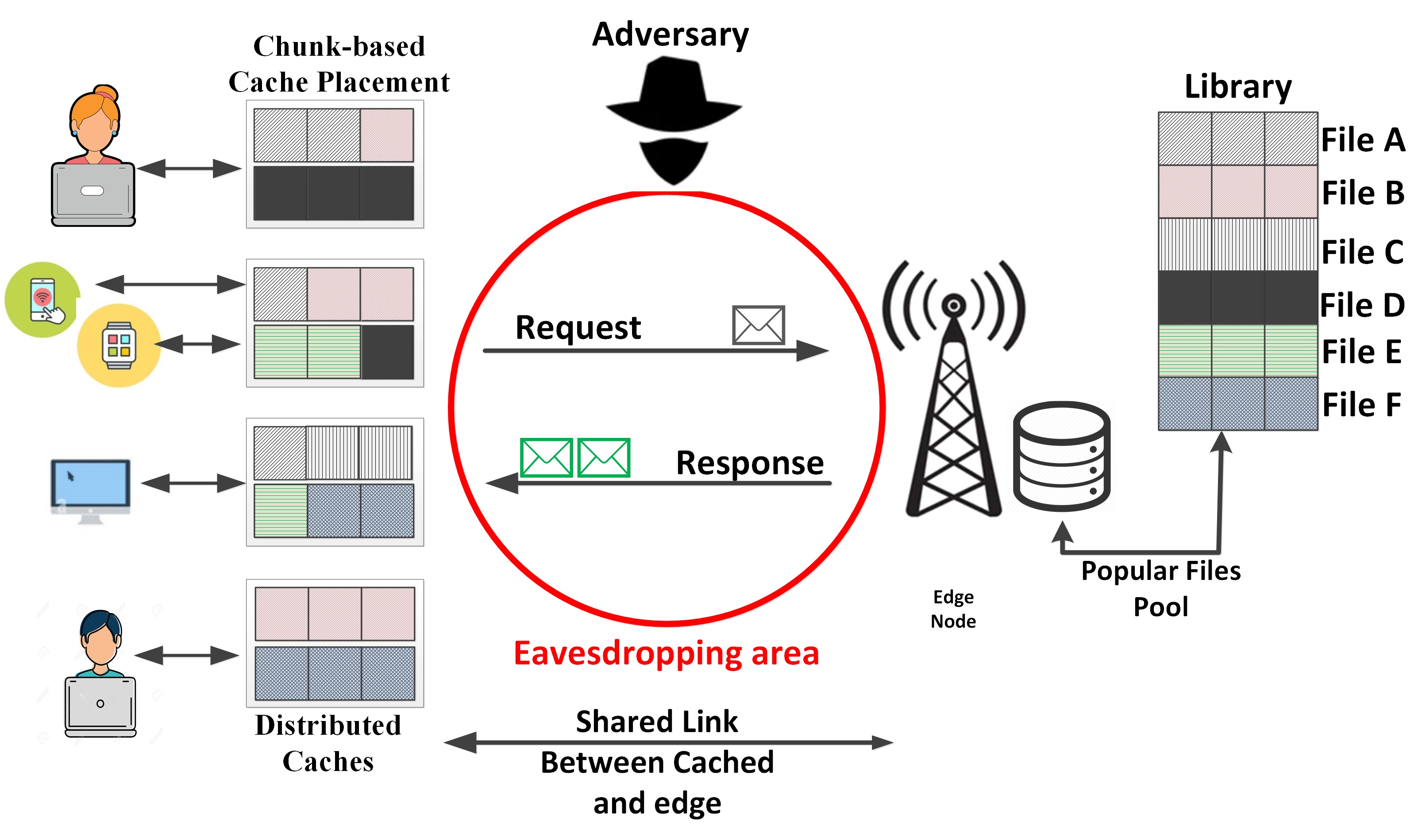}
	\caption{The adversarial system model with an edge-node which is responsible to serve the caches through a shared link.}
	\label{fig:SystemModel}
\end{figure}
In this section, we introduce our system model, an EC network comprised of an edge server node and distributed caches and an adversary eavesdropping on transmissions within the caching network. Then, we introduce the related performance metrics, including communication cost and privacy degree, and give an example to clarify the trade-off between these two metrics. Table \ref{tab:symbols} lists the notations used in this paper.

\subsection{Edge caching model}
\label{Subsec:EdgecahingModel}
As depicted in Fig.~\ref{fig:SystemModel}, our proposed caching network consists of a single edge server node, referred to as the server hereafter, and a set of distributed edge caches, denoted by $\mathcal{K}=\{1,2,\dots, K\}$, that are connected to the server through a shared link. The server has full access to a library consisting of $N$ files of equal size, i.e., $\mathcal{N}=\{1,2,...,N\}$, each of which is partitioned into $C$ chunks. Moreover, each cache $k\in\mathcal{K}$ has the capacity of $M$ files, equivalent to $MC$ chunks. 
Furthermore, each cache serves a distinct set of end-users. If the number of end-users at each cache is equal to one, the model corresponds to the case where the distributed caches are located at end-users. Otherwise, the edge caches model the distributed caches owned by an ISP. Nevertheless, both cases can be applied in our scenario since as will be discussed, we focus on the privacy of transmissions between the caches and the server, which is independent of the number of users under the coverage of caches.

In compliance with the convention, the caching process in our considered scenario is performed in two phases; cache content placement (CCP) and content delivery (CD). In the CCP phase, the server fills up the cache's storage with chunks chosen from different files. In our analysis, all chunks of a file are equally important, and thus, we only focus on the number of chunks cached from a file, regardless of which exact chunks have been cached. Moreover, the total number of chunks in each cache does not exceed the capacity of the cache, i.e., $MC$. Consequently, let $\mathbf{z}=(z_1,z_2,...,z_N)$ denote a feasible chunk placement at a typical edge cache, where $z_i$ denotes the number of chunks stored from file $i$. Then, regarding that all caches are identical, the set of feasible chunk placements at each edge cache, represented by $\mathcal{F}$, is written as
\begin{align}
\label{eq:def1feas}
\mathcal{F}=&\{\mathbf{z}|0\le z_i \le C, i\in \{1,\dots,N\}, \sum_{i=1}^{N}z_i=MC\}.
\end{align} 
 
We assume that a central entity decides about the contents of the caches in the CCP phase. Without loss of generality, we consider the server as a decision-maker in our scenario, e.g., the server plays the role of a content provider in the real world. In particular, the server chooses the placement of each cache according to a probabilistic caching policy, where the policy in fact, indicates a probability distribution over all feasible cache placements, i.e., $\mathcal{F}$. Thus, it is evident that in the proposed probabilistic caching, the chunk places in each cache are filled jointly rather than independently. In this regard, we define the joint probabilistic caching policy more rigorously as follows.

\begin{definition}{(Joint Probabilistic Caching Policy)\\}
	\label{def1}
	The joint probabilistic caching (JPC) policy at cache $k$ is defined as a probability distribution over $\mathcal{F}$, denoted by $P^{(k)}(\mathbf{z})$. Let $\mathbf{Z}^{(k)}=(Z^{(k)}_1,Z^{(k)}_2,...,Z^{(k)}_N)$ be a random vector indicating the chunk placement at cache $k$, where $Z^{(k)}_i$ is the random variable denoting the number of chunks cached from file $i$ at cache $k$. Then, $P^{(k)}(\mathbf{z})=\Pr\{\mathbf{Z}^{(k)}=\mathbf{z}\}$.
	
\end{definition}

We assume that one request is generated at each time slot in the CD phase. Also, the generated request belongs to cache $k$ with probability $p_g^{(k)}$, where $\sum_{k=1}^K p_g^{(k)}=1$. It is worth noting that the request probabilities $p_g^{(k)}$ can be used to model the influence of the number of users under the coverage of different caches and their activities.
On the other hand, each request is related to file $i$ with probability $p_i$, regardless of its generating cache. $p_i$ is referred to as the popularity of file $i$ and we have $\sum_{i=1}^{N} p_i =1$. Regarding that the popularity of the file and the requesting cache are independent, the probability that cache $k$ requests file $i$ is written as:
	\begin{equation}
	\label{eq:def2}
	P(k, i)=p_i \: p_g^{(k)}.
	\end{equation}

When cache $k$ queries file $i$, it requests for those chunks of file $i$ that are not available in its cache. Then, the server, being aware of the chunk placements at the caches, sends the un-cached chunks of file $i$. As such, the server transmits $Y^{(k)}_i=C-Z^{(k)}_i$ number of chunks back to cache $k$ over a shared link channel, where $Y^{(k)}_i$ is the random variable denoting the number of chunks transferred over the link given that cache $k$ has requested file $i$. The next part describes the adversary model and its related assumptions.

\begin{table}[!t]
	\caption{Notations}
	\begin{center}
		\begin{tabular}{|cp{6.8cm}|cp{3.0cm}|}
			\hline
			Notation & Description \\
			\hline
			$N$ &Number of files in the library\\
			$K$ &Number of caches\\
			$M$ &Cache size\\
			$\zeta$ &Guaranteed value for $\Psi$\\
			$C$  &Number of chunks in a file \\
			$q_i^{(k)}$ &Probability that  cache $k$ stores  file $i$\\
			$p_i$ &Popularity probability of  file $i$\\
			$P_e$  &Error probability of ADV\\
			$\Psi$ &Privacy degree\\
			$\Omega$ &Total communication cost\\
			$p_g^{(k)}$ &Request generation probability of cache $k$\\
			$I$ &\Gls{RV} that shows the index of the requested file\\
			$U$ &\Gls{RV} that shows the index of the requesting user \\
			$\hat{i}$ &Estimated file index\\
			$\hat{k}$ &Estimated cache index\\ 
			$Z$ &\Gls{RV} that shows the number of stored chunks\\ $Y$ &\Gls{RV} that shows the number of non-cached chunks\\
			\hline
		\end{tabular}
	\end{center}
	\label{tab:symbols}
	\vspace{-0.28in}
\end{table}

\subsection{Adversary Model}
\label{subsec:Adversary Model}
The adversary is assumed to be a passive attacker in our system model, i.e., it does not decrypt or corrupt the communication over the shared link, e.g., due to solid encryptions applied on the transferred files. The adversary in our model only eavesdrops on the communication between the server and edge caches such that it counts the number of chunks transferred in response to every request. Then, it estimates the requested file and the requesting cache based on the observed number of transferred chunks. Also, we assume that the adversary is protocol-aware. More preciously, the adversary has complete knowledge about the probabilistic caching policies, $P^{(k)}(\mathbf{z})$, the files' popularities, $p_i$, and the probability of generating requests by the caches, $p_g^{(k)}$. This assumption corresponds to practical situations where the adversary is an authorized party in the network (honest but curious). Thus, it has full access to the network information broadcast by the server in the initiation phase, e.g., the adversary can be a code running on one of the edge-caches, which are authorized components in the network \cite{hassanpour2020context,roman2018mobile}.

Although the adversary knows the probabilistic caching protocols of different caches, it does not know which specific chunks are stored at each chunk. Thus, after observing the number of chunks transferred over the link, in response to one request, the adversary estimates the requested file and the requesting cache through employing the best estimation strategy, i.e., the MAP rule, with prior knowledge of $P^{(k)}(\mathbf{z})$, $p_i$, and $p_g^{(k)}$.

Let $\hat{k}(y)$ and $\hat{i}(y)$ denote the adversary's estimation of the requesting cache and requested file, respectively, given that $y$ chunks are transferred over the shared link. Then, according to the MAP rule, $(\hat{k}(y),\hat{i}(y))$ are derived as
\begin{align}
\label{PD}
(\hat{k}(y),\hat{i}(y)) =& \argmax_{k\in\mathcal{K},i\in\mathcal{N}} P(k,i|Y=y) \\
\label{eq:PD1}
=& \argmax_{k\in\mathcal{K},i\in\mathcal{N}} P(Y=y|k,i)p_i p_g^{(k)},
\end{align}
where $Y$ is the random variable denoting the number of chunks transferred over the shared link, and $P(k,i|y)$ is the probability that $k$ has requested $i$, given that $Y=y$. Moreover, \eqref{eq:PD1} is written using the Bayes' rule and \eqref{eq:def2}.

\subsection{Communication Cost}
\label{sec:CC}
Communication cost denoted by $\Omega$, is defined as the average number of files transferred over the shared link in response to one request at the CD phase, i.e., $\Omega=\frac{1}{C}\text{E}[Y]$. Therefore, using \eqref{eq:def2}, $\Omega$ is formulated as
\begin{equation}
\hspace*{-0.2cm}\Omega=\frac{1}{C}\sum_{k=1}^{K}\sum_{i=1}^{N}P(k, i) \text{E}[Y^{(k)}_i]= \frac{1}{C}\sum_{k=1}^{K} \sum_{i=1}^{N}p_g^{(k)} p_i E[Y^{(k)}_i],
\label{eq:CC}
\end{equation}
\vspace{-0.1cm}
where as noted before $Y^{(k)}_i$ is the random variable denoting the number of transferred chunks, given that cache $k$ requests file $i$. Moreover, $E[.]$ is the expectation operator.

\subsection{Privacy Degree}
\label{sec:PD}

As discussed in Section \ref{subsec:Adversary Model}, upon observing the number of chunks transferred over the shared link, the adversary uses MAP rule, as expressed in \eqref{PD}, to estimate the requesting cache and the requested file indices, i.e., $\hat{k}(y)$ and $\hat{i}(y)$, respectively. The error happens when the adversary does not correctly detect the requested cache or the requested file. Let us define $P_{e|y}$ as the error probability of the adversary given that $y$ chunks are observed. Then, $P_{e|y}$ is written as   
\vspace*{-0.2cm} 
\begin{equation}
P_{e|y} = \mathrm{Pr}\{ U \ne \hat{k}(y) ~\text{or}~I\ne \hat{i}(y)|Y=y\},
\label{eq:Pe}
\end{equation}
where $U$ and $I$ are random variables denoting the real requesting cache and requested file, respectively. Now, the privacy degree, denoted by $\Psi$, is defined as the error probability of the adversary and is derived as 
\vspace*{-0.3cm} 
\begin{equation}
\Psi= \sum_{y=0}^{C} P_{e|y} \Pr\{Y=y\}.
\label{eq:pd0}
\end{equation}

\subsection{Example}

Here, we bring a simple example in order to clarify the motivation behind using probabilistic caching for the aim of privacy-preserving. We compare two deterministic and probabilistic caching policies in our proposed scenario with the following parameters; the library contains two files A and B with popularities $p_A = 0.8$ and $p_B = 0.2$, there is one cache with capacity one file, i.e., $k =1$ and $M =1$, and the files are not divided in smaller portions, i.e., $ C = 1$. The deterministic caching scenario caches file A with probability one. According to the assumptions in Section \ref{subsec:Adversary Model}, the adversary is aware of the caching protocol, thus, knows that file A is cached. Consequently, it infers that files A and B are requested upon observing zero and one transferred file over the shared link. This leads to an error-free detection at the adversary. On the other hand, in the probabilistic caching scenario, files A and B are chosen to be cached with probabilities 0.7 and 0.3, respectively. Then, it can be seen that $\Pr(Y = 0|A)p_A=0.7 \times 0.8 > \Pr(Y=0|B) p_B=0.3 \times 0.2$ and $\Pr(Y = 1|A)p_A=0.3 \times 0.8 > \Pr(Y=1|B) p_B=0.7 \times 0.2$, which regarding \eqref{eq:PD1}, implies that the adversary always estimates file $A$ as the requested file upon observing either zero or one file, over the shared link. Thus, the error happens when file $B$ is requested, with probability $0.3$. Therefore, through applying probabilistic caching, privacy degree increases. However, this is at the cost of increasing the communication cost. As such, in deterministic policy the communication cost $\Omega$ is equal to $\Omega=p_B=0.2$, while in the deterministic case, we have $\Omega = 0.3p_A + 0.7p_B = 0.35$. 
  
Thus, when applying the probabilistic caching, a trade-off exists between the communication cost and the privacy degree. In this paper, we characterize this trade-off. In particular, in the next section, we minimize the communication cost over all probabilistic caching policies such that the privacy degree exceeds a desired threshold and show that the formulated optimization can be written as a Linear Programming (LP) optimization.

\section{Problem Formulation}
\label{sec:problemformulation}
In this section, we optimize the probabilistic caching policies at different caches in order to minimize the communication cost while keeping the privacy degree greater than a predefined threshold. Our goal is to solve the following optimization problem
\vspace*{-0.2cm}
\begin{equation}
\begin {aligned}
\label{eq:main_opt}
&\min_{{P^{(k)}}(\mathbf{z}),\atop\scriptstyle \mathbf{z}\in \mathcal{F}, k\in\mathcal{K}} \hspace{1.0cm} \Omega \\
&\;s.t. \quad \Psi \ge \zeta,\\
& \; \qquad \sum_{\mathbf{z}\in \mathcal{F}} P^{(k)}(\mathbf{z})=1, ~~~~~~k\in\mathcal{K},\\
& \; \qquad 0 \leq P^{(k)}(\mathbf{z})\leq 1, ~~~~~~\mathbf{z}\in \mathcal{F}, ~k\in\mathcal{K},
\end{aligned}
\end{equation}
where $\Omega$ and $\Psi$ are derived in \eqref{eq:CC} and \eqref{eq:pd0}, respectively. Also, $\zeta$ denotes the threshold associated with the privacy degree. Moreover, the second and third constraints in \eqref{eq:main_opt} assures $P^{(k)}(\mathbf{z})$ be a probability mass function. In the following, we first rewrite $\Omega$ and then $\Psi$ in terms of the probabilistic caching policies, i.e., $P^{(k)}(\mathbf{z})$'s.
$E\big[Y^{(k)}_i\big]$ in \eqref{eq:CC} can be written as
\vspace{-0.2cm}
\begin{align}
\label{exp}
E\big[Y^{(k)}_i\big]&= \sum_{y=0}^{C} y\Pr\{Y^{(k)}_i=y\}=\sum_{y=0}^{C} y \Pr\{Z^{(k)}_i=C-y\}.
\end{align}

In order to calculate $\Pr\{Z^{(k)}_i=C-y\}$, we need to consider any placement $z \in \mathcal{F}$ that contain exactly $C-y$ chunks of $i$. Thus, $\Pr\{Z^{(k)}_i=C-y\}$ is written as

\begin{equation}
\Pr\{Z^{(k)}_i=C-y\}= \sum_{\mathbf{z} \in \mathcal{F},\atop\scriptstyle z_i=C-y} P^{(k)}(z).
\label{eq:A}
\end{equation}
Using \eqref{exp} and \eqref{eq:A} in \eqref{eq:CC}, the communication cost is derived in terms of $P^{(k)}(\mathbf{z})$'s as
\vspace{-0.28cm}
\begin{equation}
\Omega=\frac{1}{C}\sum_{k=1}^{K} \sum_{i=1}^{N}   \sum_{y=0}^{C}p_g^{(k)} \: p_i \; y \sum_{z \in \mathcal {F},\atop\scriptstyle z_i=C-y}  P^{(k)}(\mathbf{z}).
\label{eq:CC_2}
\end{equation} 

Next we derive the privacy degree in terms of $P^{(k)}(\mathbf{z})$'s. From \eqref{eq:Pe} and \eqref{eq:pd0}, $\Psi$ is written as
\vspace{-0.1in}
\begin{subequations}
	\begin{align}
	\label{eq:pd1}
	&\Psi=\sum_{y=0}^{C} \mathrm{Pr}\{ U \neq \hat{k}(y) ~\text{or}~ I\neq \hat{i}(y) |Y=y\} \Pr\{Y=y\}\\
    \label{eq:pd2}	
	&=\sum_{y=0}^{C} (1-\mathrm{Pr}\{U = \hat{k}(y), I= \hat{i}(y)  |Y=y\}) \Pr\{Y=y\}\\
	\label{eq:pd3}	
	&=\sum_{y=0}^{C} (1-\max_{k,i} P( k,i|Y=y))  \Pr\{Y=y\}\\
	\label{eq:pd4}
	&= 1- \sum_{y=0}^{C}\max_{k,i} \: p_i\; p_g^{(k)} \; P(Y=y|k,i) ,
	\end{align}
	\label{eq:pd}
\end{subequations}
where \eqref{eq:pd2} is replaced with \eqref{eq:pd3} using the MAP rule in \eqref{PD}. Moreover, the argument of the maximization in \eqref{eq:pd4} is written using the Bayes' rule. Using $\Pr\{Y^{(k)}_i=y\}=\Pr\{Z^{(y)}_i=C-y\}$ and \eqref{eq:A}, $P(Y=y|k,i)$ in \eqref{eq:pd4}, can be written as
\vspace{-0.18in}
\begin{equation}
\begin{aligned}
P(Y=y|k,i)=\sum_{\mathbf{z} \in \mathcal{F},\atop\scriptstyle z_i=C-y} P^{(k)}(\mathbf{z}).
\end{aligned}
\label{eq:com}
\end{equation} 
Finally, by applying \eqref{eq:com} in \eqref{eq:pd4}, $\Psi$ is written in terms of probabilistic caching policies as 
\vspace*{-0.4cm}
\begin{equation}
\begin{aligned}
\Psi=1 -\sum_{y=0}^{C}  \max_{k,i} \hspace{0.2cm} p_i .p_g^{(k)} .\sum_{\mathbf{z} \in \mathcal{F},\atop\scriptstyle z_i=C-y} P^{(k)}(\mathbf{z}) .
\end{aligned}
\label{eq:pdfinal}
\end{equation}
In the next part, we use the derived equations for $\Omega$ and $\Psi$ to formulate the optimization problem as an LP optimization.
\subsection{Linear Programming Optimization}
\label{sec:lp_opt}
Using \eqref{eq:CC_2} and \eqref{eq:pdfinal} in \eqref{eq:main_opt}, the optimization problem is rewritten as 
\begin{equation}
\begin {aligned}
\label{eq:main_opt2}
\min_{{P^{(k)}}(\mathbf{z}),\atop\scriptstyle \mathbf{z}\in \mathcal{F}, k\in\mathcal{K}} &\hspace{0.5cm}  \sum_{k=1}^{K} \sum_{i=1}^{N}   \sum_{y=0}^{C} p_i \; p_g^{(k)}\; y \sum_{\mathbf{z} \in \mathcal {F},\atop\scriptstyle z_i=C-y} P^{(k)}(\mathbf{z})\\
s.t \hspace{0.4cm}& \quad 1-\sum^C_{y=0} \max_{k,i} p_i  p_g^{(k)} \sum_{\mathbf{z} \in \mathcal{F},\atop\scriptstyle z_i=C-y} P^{(k)}(\mathbf{z}) \ge \zeta,\\
& \quad  \sum_{\mathbf{z}\in \mathcal{F}} P^{(k)}(\mathbf{z})=1, ~~~~~~k\in \mathcal{K},\\
&  \quad 0 \leq P^{(k)}(\mathbf{z})\leq 1, ~~~~~~\mathbf{z}\in \mathcal{F}, ~k\in\mathcal{K}.
\end{aligned}
\end{equation}

Note that in \eqref{eq:main_opt2}, the objective function and all constraints except the first one are linear in terms of $P^{(k)}(\mathbf{z})$. In this regard, we use additional variables, $\Gamma_y$ ($y\in\{0,1,\dots,C\}$), to change the optimization problem into a linear one. In particular, the first constraint in \eqref{eq:main_opt2} is replaced with the following inequality 
\vspace{-0.3cm}
\begin{equation}
1-\sum_{y=0}^{C} \Gamma_y \ge \zeta.
\label{eq:PD_2}
\end{equation}
Moreover, the following constraints on $\Gamma_y$'s are added to the optimization problem 
\begin{equation}
\Gamma_y \ge p_i \; p_g^{(k)} \sum_{\mathbf{z} \in \mathcal{F},\atop\scriptstyle z_i=C-y} P^{(k)}(\mathbf{z}) \quad \forall{k,i}, y.
\label{eq:Gamma}
\end{equation}
Note that the inequalities \eqref{eq:Gamma} imply that $\Gamma_y \ge \max_{k,i} \quad p_i \; p_g^{(k)}\; \sum_{\mathbf{z} \in \mathcal{F},\atop\scriptstyle z_i=C-y} P^{(k)}(\mathbf{z}),~ \forall y.$  

Using \eqref{eq:PD_2} and \eqref{eq:Gamma}, the optimization problem turns into
\begin{equation}
\begin{aligned}
\mathcal{P}_1:~&\hspace*{-0.2cm}\min_{\Gamma_y,{P^{(k)}}(\mathbf{z}),\atop\scriptstyle 0\le y\le C,\mathbf{z}\in \mathcal{F}, k\in\mathcal{K}} \hspace{0.2cm} \sum_{k=1}^{K} \sum_{i=1}^{N}   \sum_{y=0}^{C}p_g^{(k)} \; p_i \; y  \sum_{z \in \mathcal {F},\atop\scriptstyle z_i=C-y} P^{(k)}(\mathbf{z})\\
&\hspace{0.8cm}s.t. \hspace{0.3cm}1-\sum_{y=0}^{C} \Gamma_y \ge \zeta\\
&\hspace{1.6cm}\Gamma_y \ge p_i \; p_g^{(k)} \sum_{\mathbf{z} \in \mathcal{F},\atop\scriptstyle z_i=C-y} P^{(k)}(\mathbf{z}) \; \forall{k,i,y}, \\
&\hspace{1.6cm} \sum_{\mathbf{z}\in \mathcal{F}} P^{(k)}(\mathbf{z})=1, ~~~~~~k\in\mathcal{K}\\
&\hspace{1.6cm} 0 \leq P^{(k)}(\mathbf{z})\leq 1, ~~~~~~\mathbf{z}\in \mathcal{F}, ~k\in\mathcal{K}.
\label{eq:LP_General}
\end{aligned}
\end{equation}
Note that  two optimization problems \eqref{eq:main_opt2} and \eqref{eq:LP_General} are equivalent and thus, any optimal $P^{(k)}(\mathbf{z})$ in \eqref{eq:LP_General} is the solution of the optimization problem \eqref{eq:main_opt2}. As can be seen, the optimization $\mathcal{P}_1$ in \eqref{eq:LP_General} is linear in terms of $P^{(k)}(\mathbf{z})$'s and $\Gamma_y$'s, and thus, is an LP optimization. Moreover, it is worth noting that in symmetric scenarios where caches generate the requests with almost the same probabilities, i.e., $p_g^{(k)}=p_g$, all caches have the same optimal probabilistic caching policies, and it suffices to optimize $P(z)$. 

\subsection{Hit Ratio Constrained JPC}
\label{subsec:hit-ratio-constrained-JPC}
One advantage of chunk-based caching in JPC is to increase the cache hit ratio, where the hit ratio of a specific file at cache $k$ is defined as the probability that at least one of its chunks exists in the cache. In practice, especially in the case of video-type contents, once the end-user is enjoying the directly received chunks from the cache, the cache requests the un-cached chunks, leading to less experienced delay at end-users. In this regard, we incorporate the average cache hit ratio constraint into the optimization problem $\mathcal{P}_1$, as in the following.
\begin{equation}
\sum_k p^{(k)}_g h^{(k)} \leq \beta,
\label{eq:hit_ratio_avg}
\end{equation}
where $\beta$ is a constant threshold and $h^{(k)}$ denotes the hit-ratio of cache $k$. $h^{(k)}$ is written as $h^{(k)} = \sum_{i} p_i \sum_{\mathbf{z} \in \mathcal {F},\atop\scriptstyle z_i\neq 0} P^{(k)}(\mathbf{z}),$ 

where it is assumed that the request for a file hits the cache if at least one chunk of the requested file exists in the cache. It is worth noting that if we do not apply chunk-based caching in JPC, i.e., $z_i \in\{0,C\}$ instead of $z_i \in\{0,1,2,\cdots,C\}$, from \eqref{eq:CC_2}, it can be seen that the constraint $\sum_k p^{(k)}_g h^{(k)} \geq \beta$ is equivalent to $\Omega \leq 1-\beta$.

\begin{figure}[t!]
	\centering
	\includegraphics[width=0.8\linewidth]{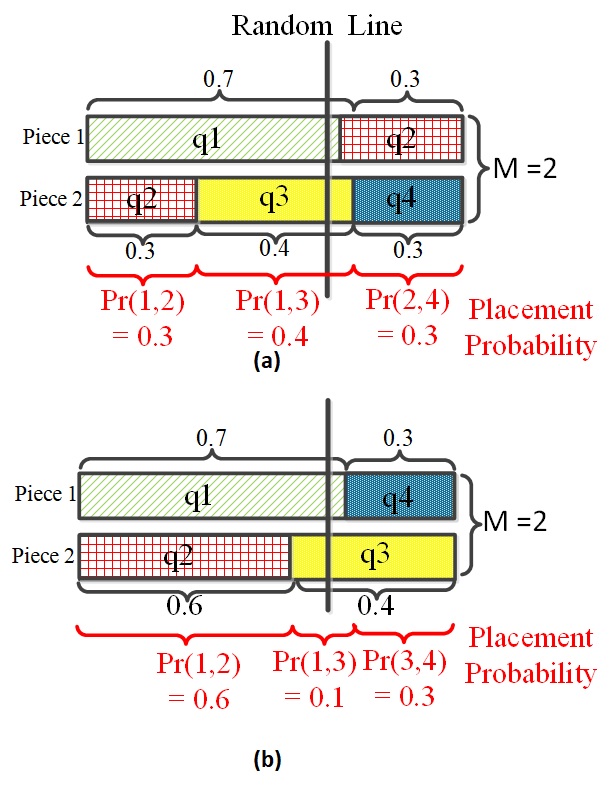}
	\caption{Cache Placement in DPC approach, with filling order of files a) $(1,2,3,4)$ b) $(1,4,2,3)$.}
	\label{fig:DPC_Approach}
\end{figure}
The number of variables in the LP optimization \eqref{eq:LP_General} is equal to $(C+1)KX$, where $X$ is the size of $\mathcal F$, i.e., the set of all possible chunk placements at caches. A bottleneck in solving \eqref{eq:LP_General} is the derivation of the set of all feasible placements, which grows explosively with the number of files, chunks, and cache sizes. Although large-scale LP optimization techniques, such as CVX, can handle large-scale problems, the complexity remains in the derivation of the set $\mathcal{F}$, mainly when the edge servers have limited energy and processing power. In this regard, we propose two probabilistic caching policies with less complexity in the rest of the paper. As such, in the next section, we optimize a non-chunk-based probabilistic caching policy introduced in \cite{blaszczyszyn2015optimal}, which reduces the number of variables to $NK$. Then, in Section \ref{sec:SPC}, we introduce the scalable version of SPC and compare its results against the optimal JPC in Section \ref{sec:evaluation_simulation}.

\section{Disjoint Probabilistic Caching (DPC) }
\label{sec:DPC}

In this section, we optimize the probabilistic caching protocol proposed in \cite{blaszczyszyn2015optimal}. In this method, the files are cached completely, thus, the set of feasible placements is
\vspace{-0.28cm} \begin{align}
\label{eq:def1feas}
\tilde{\mathcal{F}}=&\{\tilde{\mathbf{z}}| \tilde{z}_i \in\{0,C\}, i\in \{1,...,N\}, \sum_{i=1}^{N}\tilde{z}_i=MC\}.
\end{align} 

Moreover, in DPC, rather than indicating a distribution over $\tilde{\mathcal{F}}$, the probability of caching each file is determined. Then, it is shown that under some conditions, there exists a distribution over $\tilde{\mathbf{F}}$ that yields the indicated probabilities. In this regard, to be consistent with JPC, we refer to this policy as disjoint probabilistic caching (DPC). More precisely, DPC is defined as follows.

\begin{definition}
	\label{def:DPC}
	In DPC, the probability of storing file $i$ at cache $k$, denoted by $\alpha_i^{(k)}$, is indicated for $\forall i\in\mathcal{N}, \forall k\in \mathcal{K}$. In order to ensure that there exists a distribution $Q^{(k)}(\tilde{\mathbf{z}})$ over $\tilde{F}$ that yields caching probabilities $\alpha_i^{(k)}$, it is required that $\sum_{i = 1}^{N} \alpha_i^{(k)} = M $. 
\end{definition}

Note that in the above definition, $Q^{(k)}(\tilde{\mathbf{z}})$ yields $\alpha^{(k)}_i$ if the equality $\sum_{\tilde{\mathbf{z}} \in\tilde {\mathcal F},\atop\scriptstyle \tilde{z}_i= C} Q^{(k)}(\tilde{\mathbf{z}})=\alpha_i^{(k)}$ holds. In order to clarify the DPC method and its placement strategy, we bring a simple example here. Assume a scenario with $N = 4$ library files and one cache with capacity of two files, i.e., $K=1$ and $M=2$. Also, the caching probabilities of the files are chosen to be $\alpha^{(1)}_1 = 0.7$, $\alpha^{(1)}_2 = 0.6$, $\alpha^{(1)}_3 = 0.4$, and $\alpha^{(1)}_4 = 0.3$, satisfying the equality $\sum_{i = 1}^{N} \alpha_i^{(1)} = 2 $. The placement strategy in DPC works as follows. The cache is divided into $M$ (2 in our example) intervals of length one, which these intervals are placed vertically under each other, as shown in figure \ref{fig:DPC_Approach}.a. Then, the intervals are filled with $\alpha^{(k)}_i$'s (at any desired order) one after another, without replacement. Note that since $\sum_i \alpha^{(k)}_i =M$, all intervals are completely covered. Afterwards, a random number is chosen uniformly in the interval $[0,1]$, and then, according to the chosen number, a vertical line passes through all $M$ intervals, which indicates the set of $M$ files to be cached. The indicated files are definitely distinct since $\alpha_i^{(k)} \leq 1 $. For more detail please refer to \cite{blaszczyszyn2015optimal}. The aforementioned strategy is illustrated in figure \ref{fig:DPC_Approach}.a, where the intervals are filled with $\alpha^{(1)}_i$'s in ascending order of their indexes. Consequently, the distribution $Q^{(k)}(\tilde{\mathbf{z}})$ will be $Q^{(k)}(\tilde{\mathbf{z}} = (1,2)) = 0.3$, $Q^{(k)}(\tilde{\mathbf{z}}  = (1,3)) = 0.4$ and $Q^{(k)}(\tilde{\mathbf{z}}  = (2,4)) = 0.3$, and zero, otherwise.

When optimizing DPC, we optimize the probability of caching individual files, i.e., $\alpha_i^{(k)}$. Then, we find a corresponding distribution over $\tilde{\mathcal{F}}$, using the DPC caching strategy explained above. In particular, the DPC optimization is written as
\vspace*{-0.3cm}
\begin{equation}
\begin {aligned}
\label{eq:opt_proposed_policy}
\min_{\alpha^{(k)}_i: k\in\mathcal{K}, i\in \mathcal{N}}& \hspace{0.5cm} \Omega \\
s.t.\hspace{0.5cm} 
&\Psi \ge \zeta,\\
& 0\le \alpha_i^{(k)} \le 1 \quad \forall{k,i},\\
&\sum_{i=1}^{N}\alpha_i^{(k)}=M, 
\end{aligned}
\end{equation}
where the last constraint ensures that a distribution over $\tilde{\mathcal{F}}$ can be found to satisfy the file caching probabilities $\alpha_i^{(k)}$. Moreover, the communication cost ($\Omega$) is written in terms of $\alpha_i^{(k)}$'s as follows:
\vspace*{-0.35cm}
\begin{equation}
\Omega = 1 - \sum_{k=1}^{K}\sum_{i = 1}^{N} p_g^{(k)} p_i \alpha_i^{(k)},
\end{equation}
where the second term in the above equation indicates the probability that no file is transferred as the response to a typical request. Also, the privacy degree is derived from \eqref{eq:pd4}, except that the summation has two terms, corresponding to values $y = 0$ and $y = C$, respectively. Thus, the privacy degree is written as: 
\vspace{-0.28cm}
\begin{equation}
\Psi = 1 - \max_{i,k} p_g^{(k)} p_i \alpha_i^{(k)} - \max_{i,k} p_g^{(k)} p_i (1-\alpha_i^{(k)})
\label{eq:PD_1}
\end{equation}
where the first and second maximizations correspond to $y = 0$ and $y = C$, respectively. Using the same procedure as in Section \ref{sec:lp_opt}, we can turn the optimization problem in \eqref{eq:opt_proposed_policy} into an LP as follow.
\vspace{-0.28cm}
\begin{equation}
\begin {aligned}
\label{eq:opt_proposed_policy2}
\mathcal{P}_2:~\min_{\alpha^{(k)}_i: k\in\mathcal{K}, i\in \mathcal{N}}& \hspace{0.5cm} \sum_{k=1}^{K}\sum_{i = 1}^{N} p_g^{(k)}p_i (1-\alpha_i^{(k)}) \\
s.t.\hspace{0.5cm} 
& 1 - \gamma_0 - \gamma_1 \ge \zeta,\\
&\gamma_0 \ge p_g^{(k)}p_i \alpha_i^{(k)} \forall{k,i}\\
&\gamma_1 \ge p_g^{(k)}p_i (1-\alpha_i^{(k)})\forall{k,i}\\
& 0\le \alpha_i^{(k)} \le 1 \quad \forall{k,i},\\
&\sum_{i=1}^{N}\alpha_i^{(k)}=M.
\end{aligned}
\end{equation}

In the following, we compare the performance of JPC against the DPC approach and clarify its benefits over the DPC.

\subsection{DPC versus JPC}
\label{subsec:dpc_vs_jpc}
In this part, we first prove the following Lemma.
\begin{lemma}
	\label{lemmas:1}
	The chunk-based optimization of the JPC in \eqref{eq:LP_General} has the same performance as the DPC  optimization in \eqref{eq:opt_proposed_policy}.
\end{lemma}
\begin{proof}
In order to prove that optimizing JPC and DPC result in the same optimum communication cost, in the first step, we prove that any feasible point of optimization $\mathcal{P}_1$ corresponds to a feasible point of optimization $\mathcal{P}_2$, with the same communication cost.
Suppose that $P^{(k)}(\mathbf{z})$ and $\{\Gamma_0,\Gamma_1,...,\Gamma_C\}$ indicate a feasible point of $\mathcal{P}_1$ in \eqref{eq:LP_General}. Also, let $q^{(k)}_{i,x}$ be the probability that $x$ chunks of file $i$ are cached at cache $k$ in the JPC method, i.e., $q^{(k)}_{i,x}=\sum_{\mathbf{z}:z_i=x} P^{(k)}(\mathbf{z})$. Now we show that $\{\alpha_i^{(k)},\gamma_0,\gamma_1\}$, defined as 
\begin{equation} 
\alpha^{(k)}_i=1-\frac{1}{C}\sum^C_{y=0} y q^{(k)}_{i,C-y}.
\label{eq:lem1}
\end{equation}
\vspace*{-0.3cm}
\begin{equation}
\gamma_0=\frac{1}{C} \sum^C_{y=0} (C-y) \Gamma_y,~~\gamma_1=\frac{1}{C}\sum^{C}_{y=0} y \Gamma_y,  ~~
\label{eq:lem}
\end{equation}  
is a feasible solution of optimization $\mathcal{P}_2$, with the same communication cost as that of $P^{(k)}(\mathbf{z})$ in $\mathcal{P}_1$. The latter is obvious through observing that the communication costs in $\mathcal{P}_1$ and $\mathcal{P}_2$ are written as $\sum_{k} \sum_{i}  p^{(k)}_g p_i (1-\alpha^{(k)}_i)$ and $\sum_{k} \sum_{i}  p^{(k)}_g p_i \sum^{C}_{y=0} y q^{(k)}_{i,C-y}$, respectively, and thus, are equal according to \eqref{eq:lem1}. Now to show that $\{\alpha_i^{(k)},\gamma_0,\gamma_1\}$ is a feasible solution of $\mathcal{P}_2$, note that the first inequality in $\mathcal{P}_2$ is equivalent to the first inequality in $\mathcal{P}_1$ since from \eqref{eq:lem}, we have $\gamma_0+\gamma_1=\sum^C_{y=0}\Gamma_y$. Moreover, from the second constraint in $\mathcal{P}_2$, which is rewritten in terms of $q^{(k)}_{i,x}$ as $\Gamma_y \geq p^{(k)}_g p_i q^{(k)}_{i,C-y}$, and definitions of $\gamma_0$ and $\gamma_1$ in \eqref{eq:lem1}, we conclude that $\gamma_0 \geq  p^{(k)}_g p_i \sum^C_{y=0} \frac{C-y}{C} q^{(k)}_{i,C-y}=p^{(k)}_g p_i \alpha^{(k)}_i$, and $\gamma_1 \geq  p^{(k)}_g p_i \sum^C_{y=0} \frac{1}{C} q^{(k)}_{i,C-y}=p^{(k)}_g p_i (1-\alpha^{(k)}_i)$, i.e., the second and third constraints in $\mathcal{P}_2$ also hold. Finally, it remains to prove the last inequality in $\mathcal{P}_2$. Since each placement $\mathbf{z} \in \mathcal{F}$ has exactly $MC$ chunks, the average number of chunks saved in cache $k$ under policy $P^{(k)}(\mathbf{z})$ equals $MC$, i.e., $\sum^N_{i=1}\sum^C_{x=0} xq^{(k)}_{i,x} = MC$. Using this equality and definition of $\alpha^{(k)}_i$ in \eqref{eq:lem1}, we conclude $\sum_{i=1}^{N} \alpha^{(k)}_i = M$. This completes the proof of the first step. 

In the second step, we need to show that any feasible solution of $\mathcal{P}_2$, $\{\alpha^{(k)}_i,\gamma_0,\gamma_1\}$, corresponds to a feasible solution of $\mathcal{P}_1$, $\{P^{(k)}(\mathbf{z}),\{\Gamma_y\}_{y=0}^N\}$, with the same communication cost. To show this, we define $P^{(k)}(\mathbf{z})$ as follows. We set $P^{(k)}(\mathbf{z})$ equal to zero for any placement $\mathbf{z}\in \mathcal{F}$ that caches at least one file partially, i.e., $P^{(k)}(\mathbf{z})=0$ if there exists $i$ such that $z_i \notin \{0,C\}$. Moreover, using the DPC placement strategy and caching probabilities $\alpha^{(k)}_i$, we derive a distribution over $\tilde{\mathcal{F}}$ and assign it to $P^{(k)}(\mathbf{z})$. Moreover, we set $\Gamma_0=\gamma_0$, $\Gamma_C=\gamma_1$, $\Gamma_y=0$ for $y\notin \{0,C\}$. Then, it is observed the constraint and objective functions in $\mathcal{P}_1$ and $\mathcal{P}_2$ are the same. In fact, the communication cost and constraints in $\mathcal{P}_1$ are written in terms of $q^{(k)}_{i,0}$ and $q^{(k)}_{i,C}$, which are equivalent to parameters $1-\alpha^{(k)}_i$ and $\alpha^{(k)}_i$, in $\mathcal{P}_2$. This completes the proof.

\end{proof}

The above lemma shows that caching chunks of files does not help in decreasing the communication cost while keeping the privacy degree above a specific threshold. The intuition behind the above lemma is that the privacy degree improves whenever the adversary can infer less information from the number of transferred chunks. Thus, if we cache the same number of chunks for any file that will be cached, then the number of observed chunks gives no information to the adversary if one of the cached files is requested. Therefore, it is efficient to cache either the entire chunks of a file or none, leading to the following corollary.

\begin{cor}
The feasible set $\mathcal{F}$ in optimization \eqref{eq:LP_General} can be changed to $\tilde{\mathcal{F}}$. We refer to this optimization as non-chunk-based JPC.
\end{cor}

Although both optimized non-chunk-based JPC ‌and DPC result in the exact average communication cost, the difference is between the chosen distributions over $\tilde{\mathcal{F}}$. In fact, in the DPC, the cache distribution is restricted by the caching strategy introduced in \cite{blaszczyszyn2015optimal}, i.e., in what order $q^{(k)}_i$'s are filled in cache intervals. As an example, in Figure \ref{fig:DPC_Approach}, two different fillings of cache intervals result in two different cache placement distributions. The placement with the most popular files is selected with a higher probability in the second distribution, implying a better short-term performance. 
Thus, if we aim to choose a distribution that allocates greater probabilities to the placements with more popular files, this requires investigating different filling orders of the intervals in the caching strategy of the DPC approach to find the desired one. This task becomes highly complex, especially when the number of files and cache sizes increase. However, in optimizing JPC, the LP solver, i.e., CVX, automatically chooses the distribution, and our numerical results show that the chosen distribution allocates higher probabilities to the placements that include more popular files.

Another advantage of the JPC over the DPC is that it provides the possibility of chunk-based caching, and chunk-based placement, on the other hand, could be used to increase the cache-hit ratio, as discussed in Section \ref{subsec:hit-ratio-constrained-JPC}. Caching a few chunks of a less popular file in JPC leads to the hit-ratio of that file to be equal to one, while a hit-ratio equal to one is obtained in DPC only by caching the file entirely, which in turn increases the communication cost. Also, the feasible set $\mathcal{F}$ in JPC can be refined in order to include the desirable placements. For example, assume that we need to cache at least $x$ chunks of a specific file(s) to satisfy its strict end-user delay constraints. Then, we can refine the feasible set $\mathcal{F}$ such that it includes only the placements that have at least $x$ chunk of the specified file. Overall, it can be seen that whenever the short-term performance or the role of chunks becomes more important, the JPC approach is superior to the DPC, since it provides the cache placement directly. However, the number of optimization parameters in the DPC approach is lower, and thus, is preferable once the average performance of communication cost and the privacy degree is solely important. To benefit from the advantages of SPC and cope with its complexity, in the next section, we propose a scalable version of the JPC.

\section{Scalable JPC}
\label{sec:SPC}
As mentioned in Section \ref{subsec:hit-ratio-constrained-JPC}, deriving optimal JPC may be a complex and time-consuming task, especially when the number of chunks is greater than one, since it requires the computation of all feasible chunk placements. In this section, we propose a scalable version of the JPC, in which the complexity of deriving the feasible placements decreases through grouping the files into disjoint subsets. Moreover, the performance of the proposed approach can become arbitrary close to optimal JPC by increasing the number of subsets. We refer to this new version as Scalable Probabilistic Caching (SPC). 

Suppose that $S = \{S_1,S_2,\dots, S_L\}$ is a partition of $\mathcal{N}$, i.e., $\forall l,~S_l \neq \emptyset$, $\cap_{l=1}^N S_l=\mathcal{N}$, and $S_l \cap S_k = \emptyset,~\forall l,k$. Moreover, any file in $S_l$ is more popular than any file in $S_k$, given that $l < k$, i.e., $p_i>p_j,~ \forall i\in S_k, ~ \forall j\in S_l$. Then, assuming that the popularity of files decreases with their index, we have $S_l=\{\sum_{k=1}^{l-1}|S_{k}|+1,...,\sum_{k=1}^{l-1}|S_{k}|+|S_l|\}$. Also, let $\hat{\mathbf{z}}=(\hat{z}_1,\hat{z}_2,...,\hat{z}_L)$ be a vector of length $L$, where $\hat{z}_l$ indicates the number of chunks cached from subset $S_l$. Then, the set of feasible placements, denoted by $\hat{\mathcal{F}}$, is written as
\vspace*{-0.3cm}
\begin{equation}
\hat{\mathcal{F}}=\{\hat{\mathbf{z}}=(\hat{z}_1,\dots,\hat{z}_l)|0 \leq \hat{z}_l \leq |S_l|C,~ \sum_{l=1}^{L} \hat{z}_l = MC\}. 
\label{eq:F_hat}
\end{equation}

It is worth noting that in order to choose $\hat{z}_l$ chunks from $S_l$, the chunks are chosen in a uniformly random manner, one after another without replacement.

\begin{definition}{(SPC)\\}
	\label{def:SPC}
	The scalable JPC (SPC) policy at cache $k$ is defined as a probability distribution over $\hat{\mathcal{F}}$, denoted by $O^{(k)}(\hat{\mathbf{z}})$. Let $\hat{\mathbf{Z}}^{(k)}=(\hat{Z}^{(k)}_1,\hat{Z}^{(k)}_2,...,\hat{Z}^{(k)}_N)$ be a random vector indicating the chunk placement at cache $k$, where $\hat{Z}^{(k)}_i$ is the random variable denoting the number of chunks cached from subset $i$ at cache $k$. Then, $P^{(k)}(\hat{\mathbf{z}})=\Pr\{\hat{\mathbf{Z}}^{(k)}=\hat{\mathbf{z}}\}$.
\end{definition}

\subsection{Communication Cost and PD in SPC}
Let $a_l=\sum_{i\in S_l} p_i$ be the probability that one of the files in subset $S_l$ is requested. Then, the communication cost of SPC is derived as
\vspace*{-0.28cm}
\begin{equation}
\Omega =\frac{1}{C}\sum^{K}_{k=1} \sum_{l=1}^{L}  \sum_{x=0}^{|S_l| C} p_g^{(k)} \; a_l \;  q^{(k)}_{l,x}\;(C-\frac{x}{|S_l|}),
\label{eq:cc-SPC}
\end{equation}
where $q_{l,x}^{(k)}$ denotes the probability of caching $x$ chunks from subset $S_l$ at cache $k$, and thus, is written as
\vspace{-0.28cm}
\begin{equation} 
q_{l,x}^{(k)}=\sum_{\hat{\mathbf{z}}:\hat{z}_l=x}O^{(k)}(\hat{\mathbf{z}}).
\label{eq:q_l_x}
\end{equation}
Moreover, regarding the chunk selection method in SPC, each chunk in the subset $S_l$ is chosen with probability $\frac{x}{|S_l| C}$, given that $x$ chunks is selected from subset $S_l$. Thus, the average number of un-cached chunks of any file in subset $S_l$ is equal to $C-\frac{x}{|S_l|}$, i.e., the last multiplicative term in \eqref{eq:cc-SPC}.

In order to derive privacy degree, we first derive $\Pr(Y=y|i\in S_l,k)$ as
\vspace*{-0.38cm}
\begin{equation}
P(Y=y|i\in S_l,k) = \sum_{x = C-y}^{|S_l|C} q_{l,x}^{(k)} P_{x,C-y},
\label{eq:Py_l_k}
\end{equation}  
where $P_{x,C-y}$ denotes the probability that $C-y$ chunks out of $x$ selected chunks from $S_l$ belong to the requested file $i$. Note that this probability is the same for any $i\in S_l$, since as mentioned before, all chunks are chosen uniformly at random. In fact, $P_{x,C-y}$ is derived as
\vspace*{-0.3cm}
\begin{equation}
P_{x,C-y}= \frac{ {C \choose C-y} {|S_l|C-C\choose x-(C-y)}}{{|S_l|C \choose x}},
\label{eq:P_x_y}
\end{equation}
where the nominator indicates the number of ways to choose $C-y$ chunks from the requested file and the rest $x-(C-y)$ chunks from other files. Also, the denominator shows the total number of ways of choosing $x$ chunks from $|S_l|C$ available ones. Since $P(Y=y|i\in S_l, k)$ is independent of $i$, we denote it by $P(Y=y|l, k)$ hereafter, where $P(Y=y|l,k)$ denotes the probability that $y$ chunks are transferred over the shared link given that a file from subset $S_l$ is requested. Using this fact and \eqref{eq:pd4}, the privacy degree can be written as
\vspace{-0.12in}
\begin{subequations}
\begin{align}
\label{eq:PD_SPC0} \Psi &= 1- \sum_{y=0}^{C} \max_{l,k} p_g^{(k)} \max_{i\in S_l} \Pr(Y=y|i,k) p_i,\\
\label{eq:PD_SPC1}&=1- \sum_{y=0}^{C} \max_{l,k} p_g^{(k)} p^*_l \Pr(Y=y|l,k)  
\end{align}
\end{subequations}
where $p^*_l$ indicates the popularity of the most popular file in subset $S_l$.

Since $\Omega$ in \eqref{eq:cc-SPC} and the arguments of maximization in $\Psi$ in \eqref{eq:PD_SPC1} are linear functions of probability distributions $O^{(k)}(\hat{\mathbf{z}})$, we use the same procedure as in Section \ref{sec:lp_opt} to change the optimization problem of SPC to an LP optimization as follows ($ 0\le y\le C,\mathbf{z}\in \hat{\mathcal{F}}, k\in\mathcal{K}$)
\begin{subequations}
\begin{align}
\mathcal{P}_3:&\min_{\Gamma_y,{O^{(k)}}(\hat{\mathbf{z}})} \hspace{0.1cm} \frac{1}{C}\sum^{K}_{k=1} \sum_{l=1}^{L}  \sum_{x=0}^{|S_l| C} p_g^{(k)} \; a_l \;  q^{(k)}_{l,x}\;(C-\frac{x}{|S_l|})\\
s.t. \hspace{0.3cm}&1-\sum_{y=0}^{C} \Gamma_y \ge \zeta,\label{eq:cns1}\\
&\Gamma_y \ge p^*_l \; p_g^{(k)} \sum_{x = C-y}^{|S_l|C} q^{(k)}_{l,x} P_{x,C-y}, \; \forall{k,l,y},\label{eq:cns2} \\
& \sum_{\hat{\mathbf{z}}\in \hat{\mathcal{F}}} O^{(k)}(\hat{\mathbf{z}})=1, ~~~~~~k\in\mathcal{K}\label{eq:cns3}\\
& 0 \leq O^{(k)}(\hat{\mathbf{z}})\leq 1, ~~~~~~\hat{\mathbf{z}}\in \hat{\mathcal{F}}, ~k\in\mathcal{K},
\label{eq:cns4}
\end{align}
\label{eq:SPC_LP}
\end{subequations} 
where $q^{(k)}_{l,x}$ is written in terms of $O^{(k)}(\hat{\mathbf{z}})$ as in \eqref{eq:q_l_x}. The above optimization can be solved through convex optimization tools, where its complexity depends on the number of subsets, i.e., $L$. As $L$ increases the performance becomes closer to the optimal JPC, where at $L=N$, each subset contains exactly one file and both methods result in the same optimal caching policies.
Similar to the JPC approach the following lemma is proved.

\begin{lemma}
	\label{lemmas:spc-chunk}
	Suppose that there exists a constant value $h$ such that for any $l\in \{1,\cdots,L\}$, $1 \leq h \leq |S_l|$ and $hL \geq M$, then the performance of optimization $\mathcal{P}_3$ does not improve with $C$.
\end{lemma}
\begin{proof}
Please see Appendix \ref{appendix1}.
\end{proof}
Examples of the above lemma are when $L\ge M$ and $h=1$, and when $L < M$ and the sizes of all subsets are the same, i.e., $\frac{N}{L}$. In the latter case, it is enough to choose $h\geq \frac{M}{L}$, which is possible since  $N>M$.
Although the value of $C$ does not affect the optimal communication cost, but hit-ratio increases as a result of increasing $C$ since the average number of files, cached partially from each subset, increases. 
 
The hit-ratio constraint can also be added to the optimization $\mathcal{P}_3$, as $\sum_{k}p_g^{(k)} h^{(k)} \geq \beta$, where $h^{(k)}$ is derived as 
\vspace*{-0.2cm}
\begin{equation}
\label{eq:hit_SPC}
h^{(k)}  = \sum_{l} a_l \sum_{x = 0}^{|S_l|C} q_{l,x}^{(k)} \bigg( 1- \frac{{|S_l|C-C \choose x}}{{|S_l|C \choose x}}\bigg),
\end{equation}
where last multiplicative term indicates the probability that at least one chunk from $x$ chosen chunks of $S_l$ is selected from the requested file.

\section{Evaluation and Numerical Results}
\label{sec:evaluation_simulation}
In this section, we evaluate the performance of JPC, DPC and SPC approaches under different network parameters. Furthermore, we compare the performance of the DPC and SPC with the JPC approach without and with having constraint on average hit-ratio probability. We use the CVX toolbox of MATLAb simulator for solving LP optimizations \eqref{eq:LP_General}, \eqref{eq:opt_proposed_policy2} and \eqref{eq:SPC_LP}. Moreover, we represent simulation results in order to validate our analytical approach. Unless otherwise stated, the system parameters are considered to be $N=5$, $K=2$, $C=10$ and $M=2$, with popularity set $\{0.5,0.18,0.12,0.11,0.09\}$ and request generation probability set $\{0.7,0.3\}$. 

In the following, we also present a random caching policy, called random dummy approach, as a benchmark, to compare the performance of the proposed policies. Also, we derive the interval of achievable privacy degrees in RDA, JPC, DPC, and SPC methods.

\subsection{Random Dummy Approach (RDA)} \label{subsec:dummy}
As noted before, dummy-based approaches are the baseline solutions to privacy-preserving problems\cite{lu2008pad,niu2014privacy}. In the \emph{basic dummy} approach, the $M$ most popular files are cached in each cache with probability one. However, $C$ chunks are sent in response to every request, regardless that the requested file is cached or not. Hence, the adversary does not acquire any information from the transmitted chunks, leading to maximum privacy degree in the network. In order to provide different degrees of privacy in the dummy approach, we consider a modified version, namely, the \emph{Random Dummy} Approach (RDA). Similar to the basic dummy approach, in RDA, the most popular files are cached in all caches. However, each request associated to a cached file is responded with $C$ dummy chunks with probability $s$ and zero chunks with probability $1-s$. Moreover, the requests associated to un-cached files are responded with $C$ chunks, as usual. 
Then, using \eqref{eq:pd4}, at a given $s$, the privacy degree of RDA is calculated as $1-p^{\max}_g\big((1-s)p_1+\max\{p_{M+1},s p_1\}\big)$, where $p^{\max}_g = \max_k p^{(k)}_g$, and $p_{M+1}$ represents the popularity of $\{M+1\}$-th file, which is the most popular un-cached file. In the following, we derive the interval of achievable privacy degrees by RDA, JPC, DPC,and SPC.

\emph{RDA, JPC, and DPC}: As mentioned before, in order to minimize the privacy degree, first of all, the caching policy should be deterministic so that the adversary has no ambiguity on which file has been cached. In this case, using \eqref{eq:pd4}, the privacy degree is derived as $1-p_g^{\max} p^*_{cached}-p_g^{\max} p^*_{uncached}$, where $p^*_{cached}$ and  $p^*_{uncached}$ are the popularities of the most popular cached and un-cached files, respectively. Then, in order to minimize the privacy degree, it is enough to maximize $p^*_{cached}+p^*_{uncached}$, which happens when the $M$ most popular files are cached with probability one in every cache. This results in minimum privacy degree to be equal to $\Psi^{\min}=1-p_g^{\max} p_1-p_g^{\max} p_{M+1}$. 
Policies DPC and SPC can choose the most popular placement with probability one, and thus, achieve $\Psi^{\min}$. To this end, it is enough to set $q_i^{(k)}=1$ for $i\in\{1,2,...,M\}$ and zero otherwise in DPC, and in SPC, $P^{(k)}(\mathbf{z})=1$ for $\mathbf{z}=(1,2,\cdots,M)$. In RDA, the most popular files are already cached, however, through setting $s=0$, the requests for cached files are not responded with dummy files at all, resulting in the minimum privacy degree $\Psi^{\min}$.  

On the other hand, maximum privacy degree is achieved whenever the adversary cannot achieve any information from the number of chunks transferred over the shared link. In this case, the adversary chooses the most popular file as the requested file and $\argmax_k p^{(k)}_g$ as the requesting cache, leading to a maximum privacy degree equal to $\Psi^{\max}=1-p^{\max}_g p_1$. $\Psi^{\max}$ is achieved in RDA and DPC through setting $s=1$ and $q^{(k)}_i=\frac{M}{N}, ~\forall i,k$, respectively, where the latter is written considering the constraint $\sum_i q_i^{(k)}=M$ in DPC. In JPC, any probability distribution $P^{(k)}(\mathbf{z})$ that results in the same caching probability of files leads to the maximum privacy degree $\Psi^{\max}$, e.g., the probability distribution that chooses any combination $M$ files out of $N$ with the same probability $\frac{1}{{N \choose M}}$. Thus, the achievable privacy degree by JPC, DPC, and RDA is equal to $[1-p^{\max}_g(p_1+p_{M+1}),1-p^{\max}_g p_1]$.

\emph{SPC}: The following lemma introduces the achievable privacy degrees by SPC. 
\begin{lemma}
Given subsets $\{S_1,S_2,...,S_L\}$, assume a placement $\hat{\mathbf{z}}$ such that $\hat{z}_i = |S_i| C$ for $i\in\{1,...,m-1\}$,  $0 < z_m \leq |S_m| C$, and $z_i=0$ for $i > m$. In fact, $\mathbf{z}$ is the placement which chooses all chunks of the first $m-1$ subsets completely, and the remaining chunks from subset $S_m$ such that $\sum_{l=1}^m \hat{z}_l=MC$. Then, the minimum and maximum privacy degree of SPC, denoted by $\Psi^{\min}_{\text{SPC}}$ and $\Psi^{\max}_{\text{SPC}}$, respectively, are derived as
\begin{equation}
\begin{aligned}
\Psi^{\min}_{\text{SPC}}=&1-p^*_1 p_g^{max}-\max\Big\{\frac{{|S_m|C-C \choose z_m}}{{|S_m|C \choose z_m}} p^*_M p_g^{max}, p^*_{M+1} p_g^{max}\Big\}\\&-p^*_M p_g^{max} \bigg(1- \frac{{|S_m|C-C \choose z_m-C} + {|S_m|C-C \choose z_m}}{{|S_m|C \choose z_m}}\bigg), \\
\Psi^{\max}_{\text{SPC}}=&1-p^*_1 p_g^{max}.
\end{aligned}
\label{eq:pd_min_max_spc}
\end{equation}
\end{lemma}
\begin{proof}
Please see Appendix \ref{appendix2}.
\end{proof}
It is seen from the above lemma that the minimum privacy degree in SPC is dependent on $C$.
\subsection{Numerical Results}  
\label{sec:genral_opt}
In Figure \ref{fig:DPCvsJPC_NUM}, the optimal communication costs of DPC and JPC approaches, derived from \eqref{eq:LP_General} and \eqref{eq:opt_proposed_policy2}, respectively, are plotted versus the privacy degree threshold, $\zeta$. In the case of the JPC approach, the results are plotted in the cases of $C=1$ and $C=10$. As can be seen, the number of chunks does not affect the optimal communication cost in the JPC approach, as proved in Lemma \ref{lemmas:1}. Moreover, DPC and JPC approaches result in the same optimal performance. Also, it is observed that optimum communication cost increases with $\zeta$. This is due to the fact that in order to provide higher privacy degree, the randomness in caching strategy increases, which leads to caching of more popular files less probable.

\begin{figure}[t!]
	\centering
	\includegraphics[width=\linewidth]{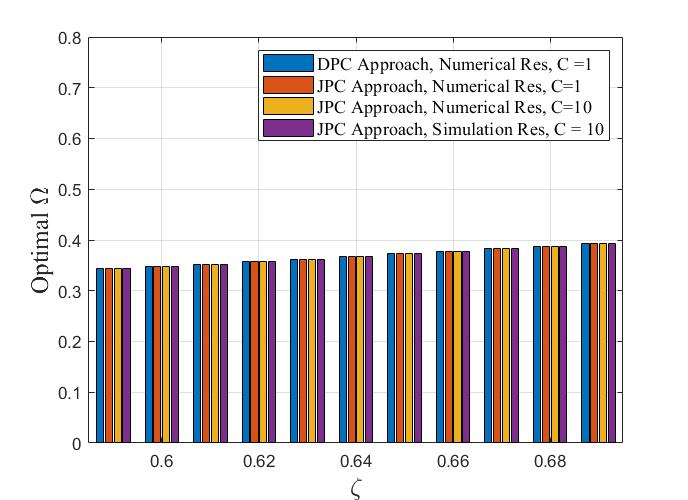}
	\caption{Numerical and simulation results of optimal $\Omega$ vs. privacy degree threshold, $\zeta$ in two approaches: DPC and JPC. }
	\label{fig:DPCvsJPC_NUM}
\end{figure}

In Fig.~\ref{fig:JPC_DPC_RDA}, the communication cost of RDA, and optimal communication costs of JPC and DPC are plotted versus $\zeta$. In order to derive the communication cost of RDA at a given value of $\zeta$, first we derive the parameter $s$ of RDA approach which results in privacy degree $\zeta$, i.e., we set $1-p^{\max}_g\big((1-s)p_1+\max\{p_{M+1},s\; p_1\}\big)=\zeta$ (see Section \ref{subsec:dummy}). Then, the communication cost of RDA is calculated as $s\sum_{i=1}^{M} p_i+\sum_{i=M+1}^{N} p_i$. 
As demonstrated in Figure \ref{fig:JPC_DPC_RDA}, all three approaches have the same communication cost at the minimum privacy degree, i.e., $\zeta=0.56$. Here, JPC and DPC cache the most popular contents, and RDA parameter $s$ is equal to zero. However, as $\zeta$ increases, the optimum communication cost increases more rapidly in the RDA approach. This is because in RDA, in order to preserve the increased privacy degree $\zeta$, we need to increase $s$, which simultaneously increases the average number of files transferred in response to all cached files. However, DPC and JPC approaches have more freedom to treat files distinctly by choosing different placements with different probabilities, thus yielding lower optimal communication cost than RDA.  Moreover, according to Figure \ref{fig:JPC_DPC_RDA}, JPC (DPC) outperforms RDA, up to 21\%.

\begin{figure} [t!]
	\centering
	\includegraphics[width=\linewidth]{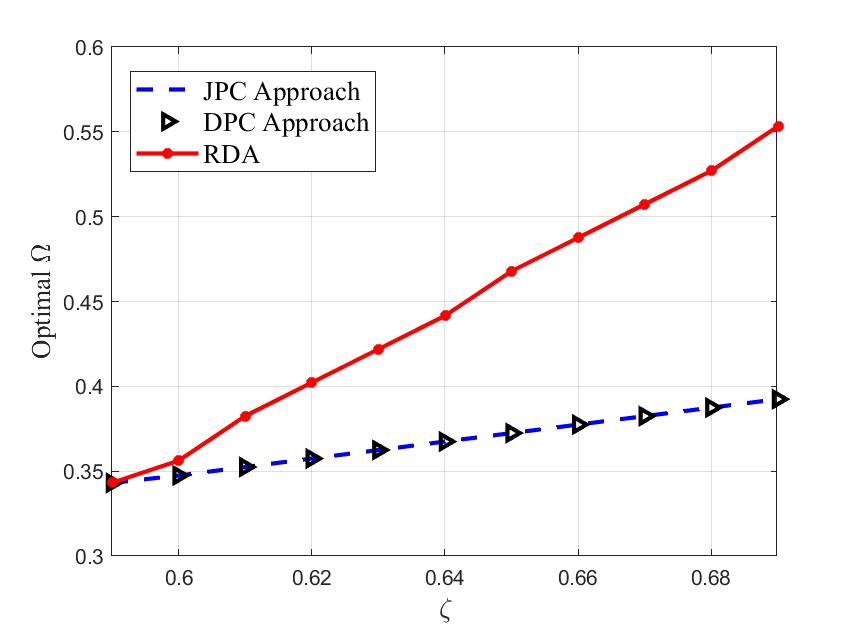}
	\caption{Optimal $\Omega$ versus privacy degree threshold, $\zeta$ when $\alpha = 1$ and $C = 1$.}
	\label{fig:JPC_DPC_RDA}
\end{figure}

Next, we investigate the impact of the number of chunks, i.e., $C$, on the hit-ratio-constrained JPC. In particular, when optimizing the hit-ratio-constrained JPC, we increase the value of $C$ until there exists a feasible solution to the optimization. We refer to the minimum value of $C$ that makes the optimization feasible as $C_{\min}$. In Figure \ref{fig:JPC_HIT}a and \ref{fig:JPC_HIT}b, $C_{\min}$ and the corresponding optimal communication cost are plotted versus the hit-ratio threshold $\beta$, respectively, assuming $N = 12$ and $M = 3$. The results are depicted for the cases that the file popularities are generated according to the Zipf distribution with parameter $\alpha =1$ and $\alpha=1.5$, and the corresponding privacy degree thresholds, $\zeta$, are considered to be $0.77$ and $0.61$, respectively. 

\begin{figure} [t!]
	\centering
	\includegraphics[width=\linewidth]{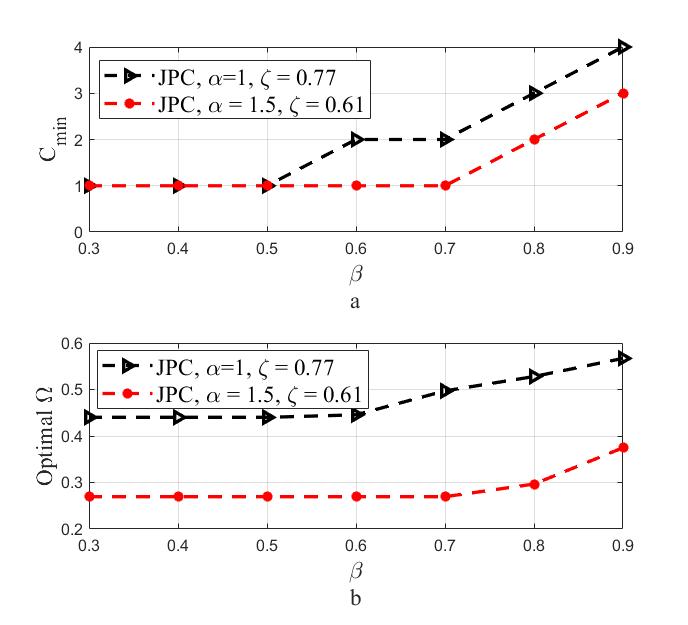}
	\caption{$a)$ $C_min$ and $b)$ optimal $\Omega$ versus hit rate threshold $(\beta)$, where $\zeta = 0.61$ and $0.77$.}
	\label{fig:JPC_HIT}
\end{figure}  
As can be observed in Figure \ref{fig:JPC_HIT}a, $C_{\min}$ increases with $\beta$ since in order to increase the number of files cached partially, we need to cache smaller parts of the files which is possible through increasing $C$. Meanwhile, the corresponding communication cost increases as observed in Figure \ref{fig:JPC_HIT}b, since by increasing the average hit ratio, smaller parts of the popular files are cached on average. 

\begin{figure}[t!]
	\centering
	\includegraphics[width=0.9\linewidth]{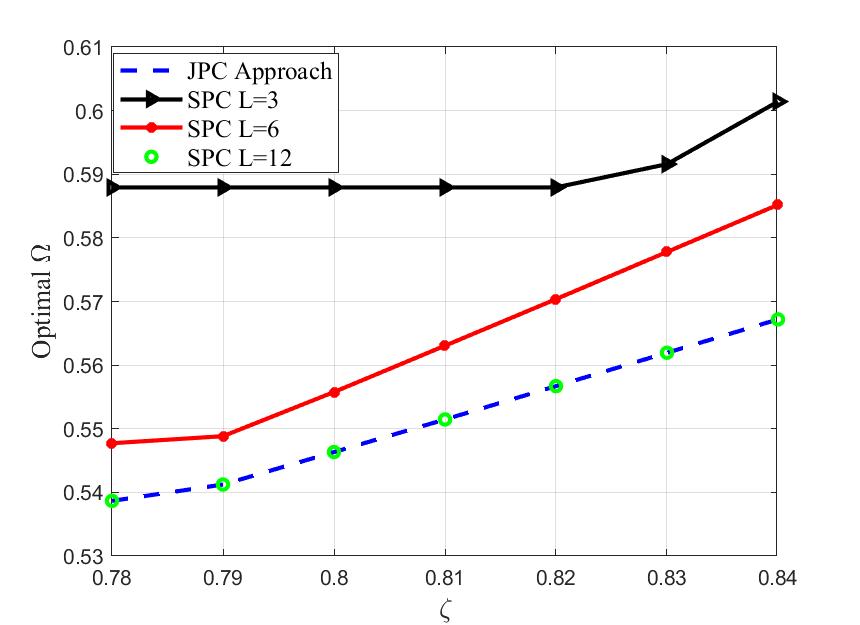}%
	\caption{Optimal $\Omega$ versus privacy degree $\zeta$ in the SPC approach, where $N=12$, $M=3$, $\alpha=0.65$.}
	\label{fig:SPC}
\end{figure}

In Figure \ref{fig:SPC}, we compare the performance of the proposed SPC approach, at different values of $L$, i.e., the number of subsets, against optimal JPC, considering Zipf parameter $\alpha=0.65$, $N=12$, and $M = 3$. Moreover, all subsets are assumed to have the same size. Also, note that the files are divided in subsets according to Section \ref{sec:SPC}, i.e., the subsets are filled with files in a descending order of popularity. As can be seen in Figure \ref{fig:SPC}, the performance of optimal SPC approach becomes closer to the optimal JPC as $L$ increase, where at $L = 12$, they have the same performance. Moreover, it is observed that with changing $L$ from 12 to 3, the optimal communication cost increases at least $6\%$ and at most $9\%$ However, the number of feasible placements decreases from 220 to 10, leading to less complex optimization problem. Another interesting point is that as we decrease $L$, the minimum achievable privacy degree increases, leading to more secure probabilistic caching strategy. As such, the minimum privacy degree increases $5\%$ at $L=3$ compared to $L=12$.

\begin{figure} [t!]
	\centering
	\includegraphics[width=\linewidth]{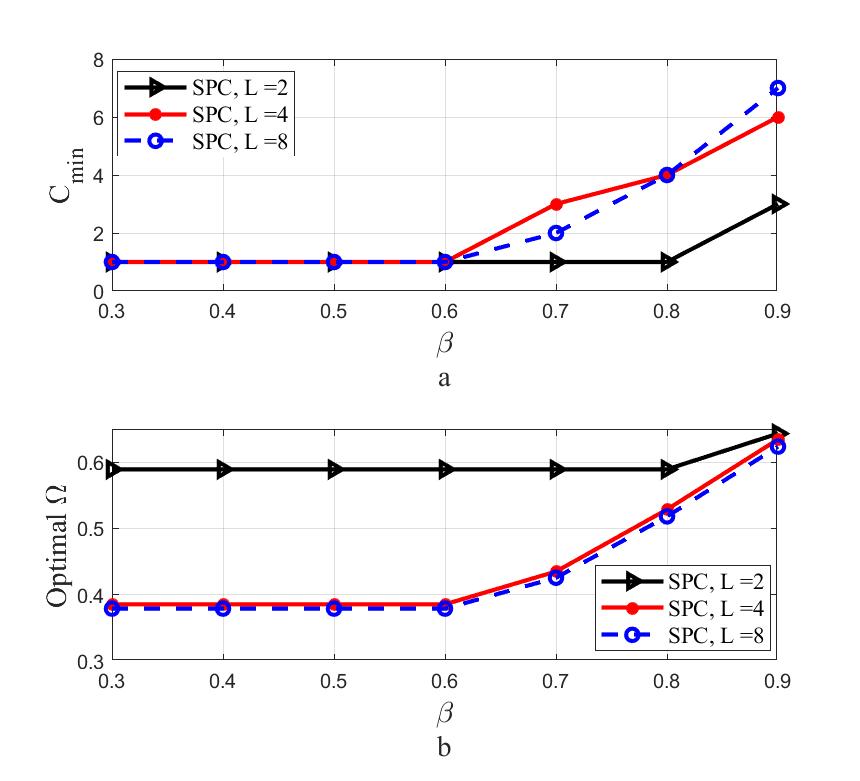}
	\caption{$a)$ $C_{\min}$ and b) Optimal $\Omega$, versus hit rate threshold ($\beta$), where $\zeta = 0.81$ and $\alpha=1$.}
	\label{fig:SPC_HIT}
\end{figure}

In Figures \ref{fig:SPC_HIT}(a) and \ref{fig:SPC_HIT}(b), $C_{\min}$, i.e., minimum value of $C$ at which hit-ratio-constrained SPC optimization is feasible, and its corresponding optimal communication cost are plotted versus hit-ratio threshold $\beta$, respectively. As observed in Figures \ref{fig:SPC_HIT}(a), $C_{\min}$ increases with $\beta$ since in order to support a higher hit-ratio, more files should be cached partially, which requires caching smaller parts of the files through increasing the value of $C$. Also, it is observed that when the number of subsets is small, generally $C_{\min}$ is smaller, e.g., at $L=3$, $C_{\min}$ increases from one at hit-ratio equal to $0.8$, however, at $L=6, 12$, such a value is equal to $0.6$. This is because, in the SPC approach, the chunks in each subset are chosen equally probably. Thus, the adversary gains no knowledge about the files within each subset and chooses the most popular file in each subset as the requested file. Consequently, when the number of files within each subset increases, or equivalently $L$ becomes smaller, the error probability of the adversary, i.e., the probability that one file other than the most popular file is requested, increases. Moreover, it is observed from \ref{fig:SPC_HIT}(b) that the optimal communication cost increases with $\beta$. This is because the space of the cache is dedicated more to less popular files to increase the hit-ratio, which in turn increases the communication cost.

\begin{figure} [t!]
	\centering
	\includegraphics[width=\linewidth]{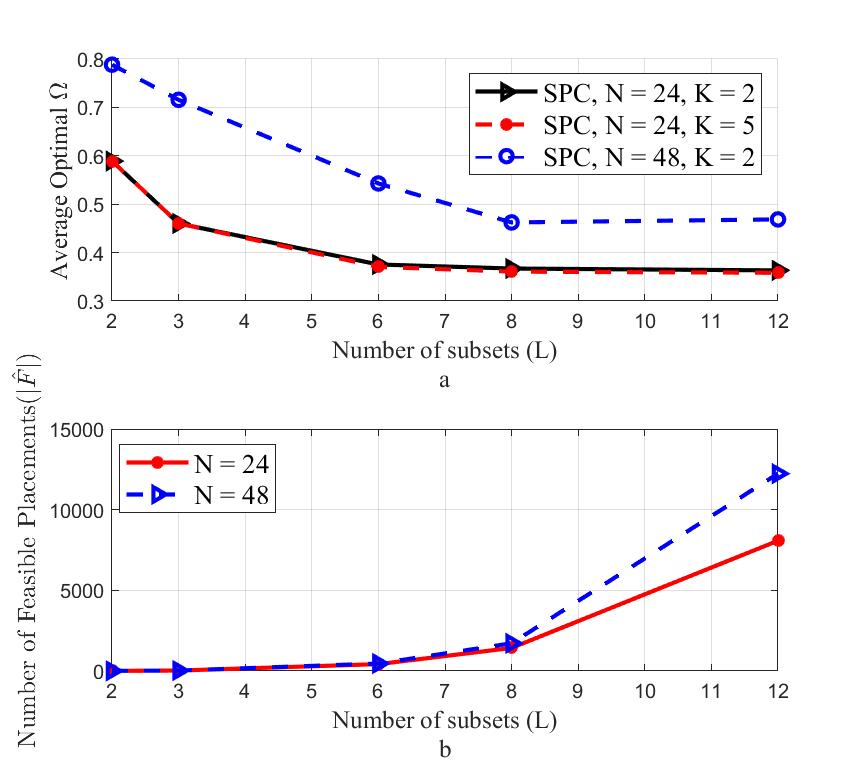}
	\caption{$a)$ Average optimal $\Omega$ and $b)$ The number of feasible placements versus number of subsets$(L)$ points, where $\alpha = 1$.}
	\label{fig:SPC_complexity}
\end{figure}

In Figure \ref{fig:SPC_complexity}(a), the average of optimal $\Omega$ over the interval of achievable privacy degrees is plotted versus the number of subsets $L$. As can be seen, the average optimal $\Omega$ decreases with $L$ since the performance of SPC becomes closer to the optimal JPC. Also, the gradient of the curves is decreasing, i.e., the difference between average optimal $\Omega$ in smaller size  $L$, e.g., $2$ and $3$ is more significant than in larger sizes of  $L$ such as $8$ and $10$. This is important because the profit ratio is not worth the computation cost despite the minimization of optimal omega in larger $L$ sizes. As depicted in Figure \ref{fig:SPC_complexity}(b), the number of feasible placements, i.e., $|\hat{\mathcal{F}}|$  grows exponentially as the $L$ size grows.

\section{Conclusion}
\label{sec:conclusion}
Decreasing the communication cost in edge networks results in caching the most popular files in edge caches. However, it degrades the users' privacy since it provides the adversary with knowledge about users' interests in different files. In this paper, we proposed JPC, based on joint chunk placement at different caches, to preserve privacy in the network while minimizing the communication cost. We showed that the corresponding optimization problem is LP. However, the LP optimization requires deriving all feasible chunk placements, which is cumbersome when the system parameters are large. To overcome this, we propose a scalable-JPC approach (SPC) in which the files are grouped into small subsets, and then the placements are done based on the subsets, i.e., the number of chunks stored from a subset. Numerical results revealed that the JPC approach outperforms SPC, DPC, and random dummy approaches.

\appendices
\section{Proof of Lemma \ref{lemmas:spc-chunk}}
\label{appendix1}
Assume that in the cases of $C=1$ and $C>1$, the SPC policy indicates the probability distributions $O^{(k)}(\hat{\mathbf{z}})$ and $\check{O}^{(k)}(\check{\mathbf{z}})$, respectively, where $\hat{\mathbf{z}}\in \hat{\mathcal{F}}$ (see \eqref{eq:F_hat}) and $\check{\mathbf{z}}=(\check{z}_1,\check{z}_2,\cdots,\check{z}_L)$ represents a feasible placement in the case of $C=1$. In fact, $\check{z}_l$ indicates the number of files chosen from subsets $S_l$. Thus, $\check{\mathbf{z}}\in \check{\mathcal{F}} = \{\check{\mathbf{z}}|\check{z}_l \in \{0,1,\cdots,|S_l|\}, 1 \leq l \leq L, \sum_{l=1}^L \check{z}_l = M \}$. In the case that $C>1$, we refer to the optimization $\mathcal{P}_3$ in \eqref{eq:SPC_LP} as $\mathcal{P}_{3c}$, and rewrite it as:
\vspace{-0.28cm}
\begin{equation}
\begin{aligned}
\mathcal{P}_{3c}:&\min_{\{\Gamma_y\}_{y=0}^C,{O^{(k)}}(\hat{\mathbf{z}})} \hspace{0.1cm} \sum^{K}_{k=1} \sum_{l=1}^{L}  p_g^{(k)} \; a_l \;  E[Y^{(k)}_l]\\
s.t. \hspace{0.3cm}&1-\sum_{y=0}^{C} \Gamma_y \ge \zeta,\\
&\hspace{0.3cm}\Gamma_y \ge p^*_l \; p_g^{(k)} P(Y=y|l,k), \; \forall{k,l,y}, 
\label{eq:SPC_C}
\end{aligned}
\end{equation} 
where from \eqref{eq:SPC_LP}, $E[Y^{(k)}_l]=\sum_{x=0}^{|S_l| C} q^{(k)}_{l,x}(1-\frac{x}{|S_l|C})$ is the average number of transferred files under policy $O^{(k)}(\hat{\mathbf{z}})$, given that cache $k$ requests one file from $S_l$. Moreover, $P(Y=y|l,k)$ in the second constraint of \eqref{eq:SPC_C} is derived from \eqref{eq:Py_l_k} and \eqref{eq:P_x_y}. Moreover, in the case of $C=1$, we rewrite $\mathcal{P}_3$ as:
\begin{equation}
\begin{aligned}
\mathcal{P}_{31}:\min_{\check{\Gamma}_0,\check{\Gamma}_1,{\check{O}^{(k)}}({\check{\mathbf{z}}})} &\hspace{0.1cm} \sum^{K}_{k=1} \sum_{l=1}^{L}  p_g^{(k)} \; a_l \;  E[\check{Y}^{(k)}_l]\\
s.t. \hspace{0.3cm}&1-\check{\Gamma}_0-\check{\Gamma}_1 \ge \zeta,\\
&\hspace{0.3cm}\check{\Gamma}_0 \ge p^*_l \; p_g^{(k)} P(\check{Y}=0|l,k), \; \forall{k,l}, \\
&\hspace{0.3cm}\check{\Gamma}_1 \ge p^*_l \; p_g^{(k)} P(\check{Y}=1|l,k), \; \forall{k,l}, 
\label{eq:SPC_1}
\end{aligned}
\end{equation} 
where $\check{Y}$ and $E[\check{Y}^{(k)}_l]$ are the corresponding values of ${Y}$ and $E[{Y}^{(k)}_l]$ in the case of $C=1$. $P(\check{Y}=y|l,k)$ and $E[\check{Y}^{(k)}_l]$ are derived similar to $P({Y}=y|l,k)$ and $E[{Y}^{(k)}_l]$ through setting $C=1$.

Now, let $\{o^{(k)}(\hat{\mathbf{z}}),\{\gamma_y\}^C_{y=0}\}$ be a feasible solution of $\mathcal{P}_{3c}$, resulting in $E[Y^{(k)}_l]=e^{(k)}_l$. Then, through the following lemmas, we show that the optimal communication cost of $\mathcal{P}_{3c}$ is less than or equal to that of $\mathcal{P}_{31}$ .
\begin{lemma}
If there exists a caching policy $\check{o}^{(k)}(\check{\mathbf{z}})$ over $\check{\mathcal{F}}$, under which $E[\check{Y}^{(k)}_l]=e^{(k)}_l$, then, $\{\check{o}^{(k)}(\check{\mathbf{z}}),\check{\gamma}_0,\check{\gamma}_1\}$, with $\check{\gamma}_0$ and $\check{\gamma}_1$ defined as $\check{\gamma}_0=\frac{1}{C}\sum^C_{y=0}(C-y)\gamma_y$ and $\check{\gamma}_1=\frac{1}{C}\sum^C_{y=0}y\gamma_y$ is a feasible solution of $\mathcal{P}_{31}$, and results in the same communication cost as $o^{(k)}(\hat{\mathbf{z}})$ in optimization $\mathcal{P}_{3c}$.
\end{lemma}

\begin{proof} 
Regarding the definition of $\check{\gamma}_1$ and the second constraint in $\mathcal{P}_{3c}$, we conclude that $\check{\gamma}_1 \geq p^*_l p_g^{(k)} \frac{1}{C}\sum_{y=0}^C yP(Y=y|l,K)= p^*_l p_g^{(k)}e^{(k)}_l$ and $\check{\gamma}_0 \geq p^*_l p_g^{(k)} (1-e^{(k)}_l)$. These two inequalities are in fact the third and second constraints in $\mathcal{P}_{31}$ since on the one hand, under any policy belonging to $\check{\mathcal{F}}$, we have $E[\check{Y}^{(k)}_l]=P(\check{Y}=1|l,k)$, leading to $e^{(k)}_l=P(\check{Y}=1|l,k)$.  The first constraint is also satisfied since $\check{\gamma}_0+\check{\gamma}_1=\sum_{y=0}^C \gamma_y$. Finally, the communication cost under both policies $o^{(k)}(\mathbf{z})$ and $\check{o}^{(k)}(\check{\mathbf{z}})$ are the same and equal to $\sum_k \sum_l p^{(k)}_g a_l e_l^{(k)}$. 
 
\end{proof}
\begin{lemma}
Suppose that there exists a constant value $h$ such that for any $l\in \{1,\cdots,L\}$, $1 \leq h \leq |S_l|$ and $hL \geq MC$, then there exists a policy $\check{o}^{(k)}(\mathbf{z})$ over $\check{\mathcal{F}}$ under which $E[\check{Y}^{(k)}_l]=e^{(k)}_l$.
\end{lemma}
\begin{proof}
We introduce the policy $\check{o}^{(k)}(\mathbf{z})$ as follows. First of all, we assume that $\check{o}^{(k)}(\check{\mathbf{z}})=0$ for any $\check{\mathbf{z}}$ that $\exists \check{z}_i\not \in \{0,h\}$, implying that either zero or $h$ files are cached from any subset $l$. Also, let  $q^{(k)}_{l,h}$ and $q^{(k)}_{l,0}$ denote the probabilities of caching $h$ and zero files from subset $S_l$ at cache $k$, under policy $\check{o}^{(k)}(\mathbf{z})$. Then, we have $E[\check{Y}^{(k)}_l]=q^{(k)}_{l,0}+q^{(k)}_{l,h}(1-\frac{h}{|S_l|})$, i.e., with probabilities of $q^{(k)}_{l,0}$ and $q^{(k)}_{l,h}(1-\frac{h}{|S_l|})$ one file is transferred in response to any file requested from $S_l$ by cache $k$. Note that $1-\frac{h}{|S_l|}$ is the probability that the requested file is not among $h$ selected files. Then from $E[\check{Y}^{(k)}_l]=e^{(k)}_l$ and $q^{(k)}_{l,h}+q^{(k)}_{0,h}=1$, $q^{(k)}_{l,h}$ is derived as $\frac{|S_l|(1-e^{(k)}_l)}{h}$. It can be seen that $\sum_{l=1}^L h q^{(k)}_{l,h}=\sum_{l=1}^L |S_l|(1-e^{(k)}_l)$. Since $|S_l|(1-e^{(k)}_l)$ is the average number of files cached from $S_l$ under policy $O^{(k)}(\hat{\mathbf{z}})$, we have $\sum_{l=1}^L |S_l|(1-e^{(k)}_l)=M$. Consequently, $\sum_{l=1}^L  q^{(k)}_{l,h}=\frac{M}{h}$. Now, we use the caching placement strategy of DPC, assuming that we have $L$ files, a cache of size $\frac{M}{h}$ and caching probabilities $q^{(k)}_{l,h}$. Regarding the equality $\sum_{l=1}^L  q^{(k)}_{l,h}=\frac{M}{h}$, a distribution over placements $\check{\mathbf{z}}\in \check{\mathcal{F}}$ is determined, where $\frac{M}{h}$ elements of each placement is equal to $h$ and other elements equal to zero. This distribution is assigned to $\check{o}^{(k)}(\mathbf{z})$.   
\end{proof}

\section{Proof of the privacy degree boundaries in SPC}
\label{appendix2}
The minimum privacy degree in the SPC is achieved when SPC tries to cache the most popular files as much as possible. Thus, the SPC chooses a placement in which all files of the subsets with lower index are chosen completely, as much as possible. Suppose that such a placement caches all files of the first $m-1$ subsets, i.e., $z_l=C|S_l|$ for $ 1\leq l \leq m-1$, and caches only $z_m < C |S_m|$ chunks from subset $m-1$. We calculate the term $\max_{l,k} p_g^{(k)} p^*_l \Pr(Y|l,k)$ in \eqref{eq:PD_SPC1}, for $y\in\{0,1,\cdots,C\}$, to calculate the minimum privacy degree. 

\emph{Case $y=0$}: $\Pr(Y=0|l,k)$ is calculated as 
\begin{equation}
\Pr(Y=0|l,k)=\begin{cases}
1;~~l\in\{1,2,\cdots,m-1\},\\
\frac{{|S_m|C-C \choose z_m-C}}{{|S_m| C \choose z_m}};~~l=m,\\
0;~~~\text{otherwise}.
\end{cases}
\end{equation}
Thus, we have
\begin{equation}
\begin{aligned}
\max_{l,k}~~ &p_g^{(k)} p^*_l \Pr(Y=0|l,k)= \max\big \{\max_{1 \leq k\leq m-1} p_g^{(k)} p^*_l,\\&p_g^{(k)} p^*_{m}\frac{{|S_m|C-C \choose z_m-C}}{{|S_m| C \choose z_m}},0 \big \}=p_g^{\max} p^*_1
\end{aligned}
\label{eq:y0}
\end{equation}

\emph{Case $y=C$}: $\Pr(Y=C|l,k)$ is calculated as
\begin{equation}
\Pr(Y=C|l,k)=\begin{cases}
0;~~l\in\{1,2,\cdots,m-1\},\\
\frac{{|S_m|C-C \choose z_m}}{{|S_m| C \choose z_m}};~~l=m,\\
1;~~~\text{otherwise}.
\end{cases}
\end{equation}
Thus, we have
\begin{equation}
\begin{aligned}
&\max_{l,k}~~ p_g^{(k)} p^*_l \Pr(Y=C|l,k)= \max\big \{0,p_g^{\max} p^*_{m}\frac{{|S_m|C-C \choose z_m}}{{|S_m| C \choose z_m}},\\ &\max_{m+1 \leq l \leq L} p_g^{(k)} p^*_l \big \}= \max \big\{p_g^{\max} p^*_{m}\frac{{|S_m|C-C \choose z_m}}{{|S_m| C \choose z_m}},p_g^{\max} p^*_{m+1} \big\}. 
\end{aligned}
\label{eq:yC}
\end{equation} 

\emph{Case $y\in\{1,\cdots,C-1\}$}: In this case, $\Pr(Y=y|l,k)=0$ for $l\neq m$ since any file requested from $m-1$ first subsets and from subsets $m+1$ to $L$ leads to $0$ and $C$ chunks transferred, respectively. Thus, for $y\neq 0,C$, we have $\max_{l,k}~~ p_g^{(k)} p^*_l \Pr(Y=y|l,k)= p_g^{\max} p^*_m \Pr(Y=y|m,k)$. Consequently, we have
\begin{equation}
\begin{aligned}
&\sum _{y=1}^{C-1} \max_{l,k}~~ p_g^{(k)} p^*_l \Pr(Y=y|l,k)\\&= p_g^{\max} p^*_m (1-\Pr(Y=0|m,k))-\Pr(Y=C|m,k))\\& = p_g^{\max} p^*_m(1-\frac{{|S_m|C-C \choose z_m}}{{|S_m| C \choose z_m}}-\frac{{|S_m|C-C \choose z_m-C}}{{|S_m| C \choose z_m}}).
\end{aligned}
\label{eq:ynot0C}
\end{equation}  
From \eqref{eq:y0}, \eqref{eq:yC}, \eqref{eq:ynot0C}, and \eqref{eq:PD_SPC1}, $\Psi_{\text{SPC}}^{\min}$ in \eqref{eq:pd_min_max_spc} can be concluded. Moreover, the maximum privacy degree in SPC is achieved when all possible placements are chosen with the same probability. In this case, the adversary chooses the most popular file as the requested file and $\argmax_k p_g^{(k)}$ as the requesting users, thus, leading to a maximum privacy degree equal to $1-p^*_1 p_g^{\max}$.

\ifCLASSOPTIONcompsoc
%
%

\ifCLASSOPTIONcaptionsoff
  \newpage
\fi



%
\bibliographystyle{IEEEtran}
\bibliography{privacy_jrnl_Spt}
\end{document}